\documentclass[sigconf,screen]{acmart}

\AtBeginDocument{%
  }

\usepackage{centernot}
\usepackage{color,soul}
\usepackage{amsmath}
\usepackage{mathtools}
\usepackage{bm}
\usepackage{scalerel}
\usepackage{subcaption}
\usepackage{enumitem}
\usepackage{cleveref}
\usepackage{svg}

\crefformat{section}{\S#2#1#3} 
\usepackage{float}
\usepackage[ruled,shortend,linesnumbered]{algorithm2e}

\usepackage{pifont}
\newcommand{\cmark}{\ding{51}}%
\newcommand{\xmark}{\ding{55}}%
\usepackage{multirow}
\usepackage{booktabs}    
\usepackage{multirow}
\usepackage{array}
\usepackage{graphicx}
\usepackage{adjustbox}
\usepackage{arydshln}

\newcommand{\db}{\textsf{DB}}

\newcommand{\resp}{\textsf{rsp}}
\newcommand{\responses}{\mathcal{R}}
\newcommand{\queries}{\mathcal{Q}}
\newcommand{\query}{q}
\newcommand{\rid}{\textsf{id}}
\newcommand{\rids}{\mathcal{I}}
\newcommand{\values}{\mathcal{V}}

\newcommand{\V}{\mathcal{V}}

\newcommand{\selector}{\textsf{Selector}}
\newcommand{\transl}{\textsf{Translator}}
\newcommand{\solver}{\textsf{Solver}}

\newcommand{\multiset}{M}
\newcommand{\support}{\textsf{T}}

\newcommand{\rs}[1]{\texttt{RS}{\scaleto{#1}{5pt}}}

\DeclareMathOperator*{\argmax}{arg\,max}
\usepackage{mathtools}
\DeclarePairedDelimiter\ceil{\lceil}{\rceil}

\newcommand{\qdist}{QD}
\newcommand{\rdist}{RD}

\newcommand{\lsetup}{\mathcal{L}_{Setup}}
\newcommand{\lquery}{\mathcal{L}_{Query}}

\newcommand{\ridp}{\textsf{rid}}
\newcommand{\qeq}{\textsf{qeq}}
\newcommand{\trlen}{\textsf{trlen}}
\newcommand{\edb}{\textsf{EDB}}

\newcommand{\Trpdr}{\mathsf{Trpdr}}
\newcommand{\Search}{\mathsf{Search}}

\newcommand{\Setup}{\mathsf{Setup}}
\newcommand{\sk}{\mathsf{sk}}

\newtheorem{remark}{Remark}

\newcommand{\framework}{\textsf{LAMa}\ }

\begin{document}

\title{How Query Distribution Knowledge Breaks Multidimensional Encrypted Range Queries, With Guarantees}

\author{Daniel Blackley}
\authornote{Equal contribution, alphabetical ordering.}
\affiliation{
 \institution{George Mason University}
 \city{Fairfax, VA, USA}
 \country{}
}
\email{dblackle@gmu.edu}

\author{Nathaniel Moyer}
\authornotemark[1]
\affiliation{
 \institution{George Mason University}
 \city{Fairfax, VA, USA}
 \country{}
}
\email{nmoyer5@gmu.edu}

\author{Charalampos Papamanthou}
\affiliation{
  \institution{Yale University}
   \city{New Haven, CT, USA}
    \country{}
}
\email{charalampos.papamanthou@yale.edu}

\author{Evgenios M. Kornaropoulos}
\affiliation{
 \institution{George Mason University}
 \city{Fairfax, VA, USA}
 \country{}
}
\email{evgenios@gmu.edu}

\renewcommand{\shortauthors}{Blackley, Moyer, Papamanthou, Kornaropoulos}

\begin{abstract}
In this work, we show how knowledge of the query distribution, combined with access-pattern leakage, is sufficient to break multi-dimensional encrypted range queries, with provable guarantees. Prior attacks either recover only data topology without concrete coordinates for plaintexts (and as a result require post-hoc transformations), or assume adversarial control over database content; a strong and unrealistic threat model. 
Given knowledge of the query distribution, we revisit \emph{frequency matching}, one of the earliest cryptanalytic ideas in this area, and push it to its limits in the multi-dimensional regime through \framework (\underline{L}eakage-\underline{A}buse via \underline{Ma}tching). \framework is a three-component framework that reconstructs plaintext coordinates in arbitrary dimensions without post-hoc transformations or data injection/poisoning. We complement \framework with the first rigorous guarantees for multi-dimensional frequency-matching cryptanalysis, covering its query complexity, optimal parameterization, and worst-case reconstruction quality. Experiments on real-world data show that \framework consistently outperforms the state of the art.
\end{abstract}

\ccsdesc[500]{Security and privacy~Cryptanalysis and other attacks}

\keywords{Encrypted Databases, Attacks, Range Queries}

\settopmatter{printacmref=false, printccs=false, printfolios=false}
\renewcommand\footnotetextcopyrightpermission[1]{}  

\maketitle

\pagestyle{empty}

\section{Introduction}

Searchable Encryption~\cite{DBLP:conf/sp/SongWP00} is a cryptographic primitive that enables efficient search over encrypted data by revealing information about the pattern of querying/accessing, known as a \emph{leakage profile}. 
The first SE scheme was introduced by Curtmola \emph{et al.}~\cite{DBLP:conf/ccs/CurtmolaGKO06}. Since then, the community has produced research covering topics such as dynamic schemes~\cite{DBLP:journals/corr/abs-2201-05006,DBLP:conf/ccs/KamaraPR12,DBLP:conf/fc/KamaraP13,DBLP:conf/ndss/CashJJJKRS14}, geometric queries~\cite{Demertzis:2016:PPR:2882903.2882911,DBLP:conf/esorics/FaberJKNRS15,DBLP:journals/popets/BoldyrevaT21,DBLP:journals/tods/DemertzisPPDGP18}, locality-aware schemes~\cite{Asharov:2016:SSE:2897518.2897562,Demertzis:2017:FSE:3035918.3064057,DBLP:conf/eurocrypt/CashT14,10.1007/978-3-319-96884-1_13}, leakage suppression~\cite{DBLP:conf/eurocrypt/GeorgeKM21,DBLP:conf/crypto/KamaraMO18,DBLP:journals/popets/Ando22}, and quantifying the privacy of SE constructions~\cite{DBLP:conf/ccs/KornaropoulosMP22,DBLP:journals/popets/BoldyrevaGW24}. 
Recently, there has been a surge in leakage attacks that  reconstruct the plaintext of databases~\cite{cryptoeprint:2019:395,DBLP:conf/ccs/GrubbsLMP18,10.1109/SP.2019.00030,10.1145/3319535.3363210,state-of-the-uniform,k-nn-attack,fmacrst-fdrtd-20,DBLP:conf/ccs/KellarisKNO16,cryptoeprint:2019:1224,cryptoeprint:2019:1175,DBLP:conf/uss/ZhangKP16,response-hiding,oya2020hiding,cryptoeprint:2021:1035,DBLP:conf/ccs/MarkatouFTS21,DBLP:journals/popets/KamaraKMDPT24}. 

    \begin{figure}[t!]
        \subfloat[Original Plaintext]{%
            \includegraphics[width=.48\linewidth,trim=4cm 3cm 1cm 2cm,clip]{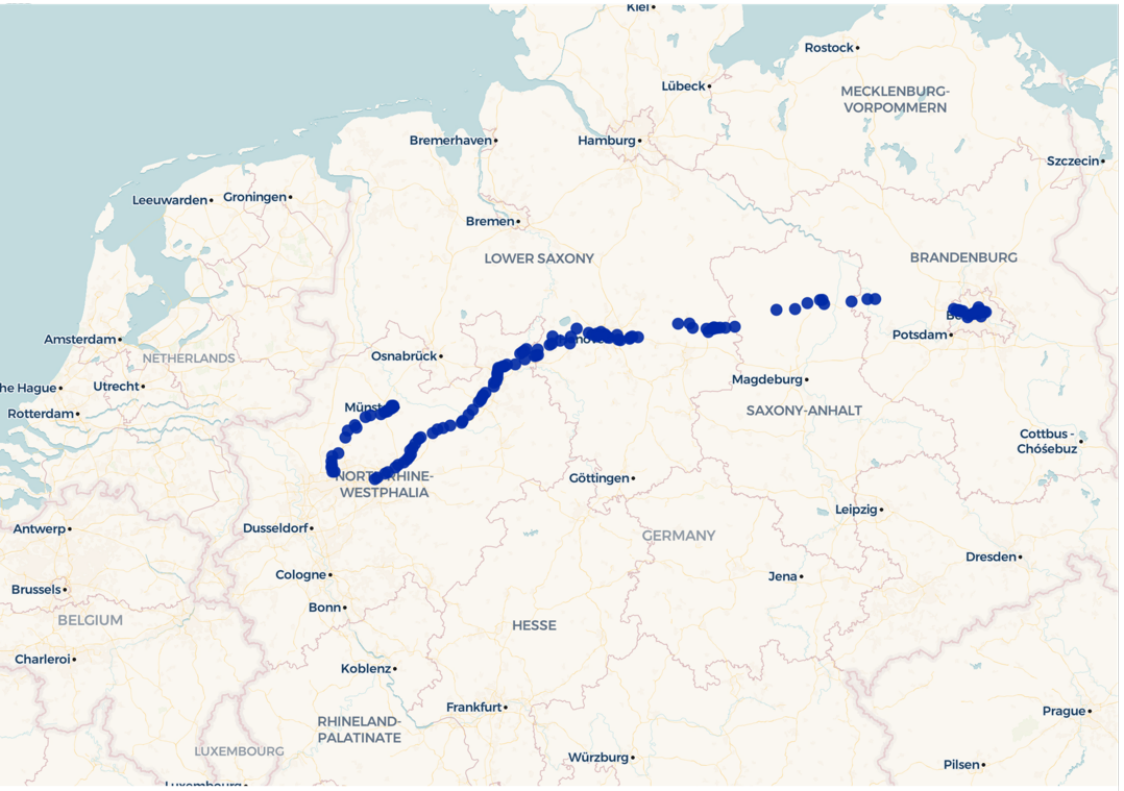}%
            \label{subfig:a}%
        }\hfill
        \subfloat[\texttt{Even Less}~\cite{even_less}, MSE=2813]{%
            \includegraphics[width=.48\linewidth,trim=4cm 3cm 1cm 2cm,clip]{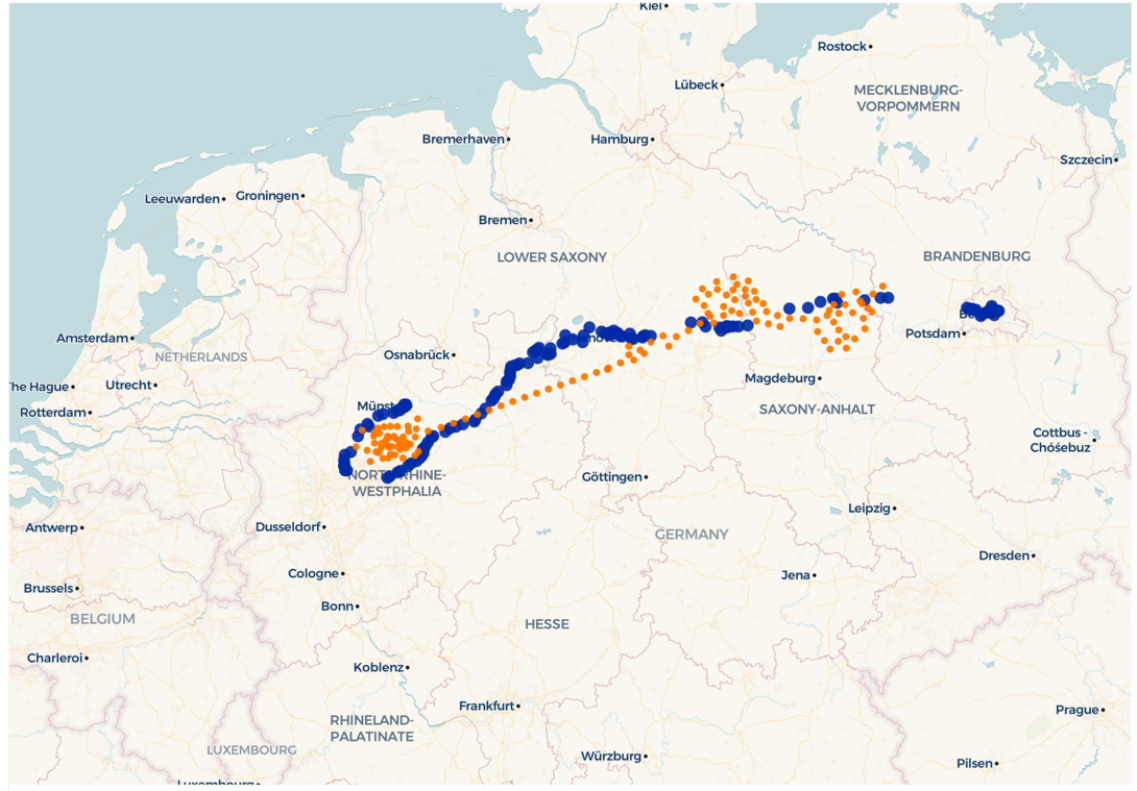}%
            \label{subfig:b}%
        }\\
        \subfloat[\texttt{Remin}~\cite{remin}, MSE=1892]{%
            \includegraphics[width=.48\linewidth,trim=4cm 3cm 1cm 2cm,clip]{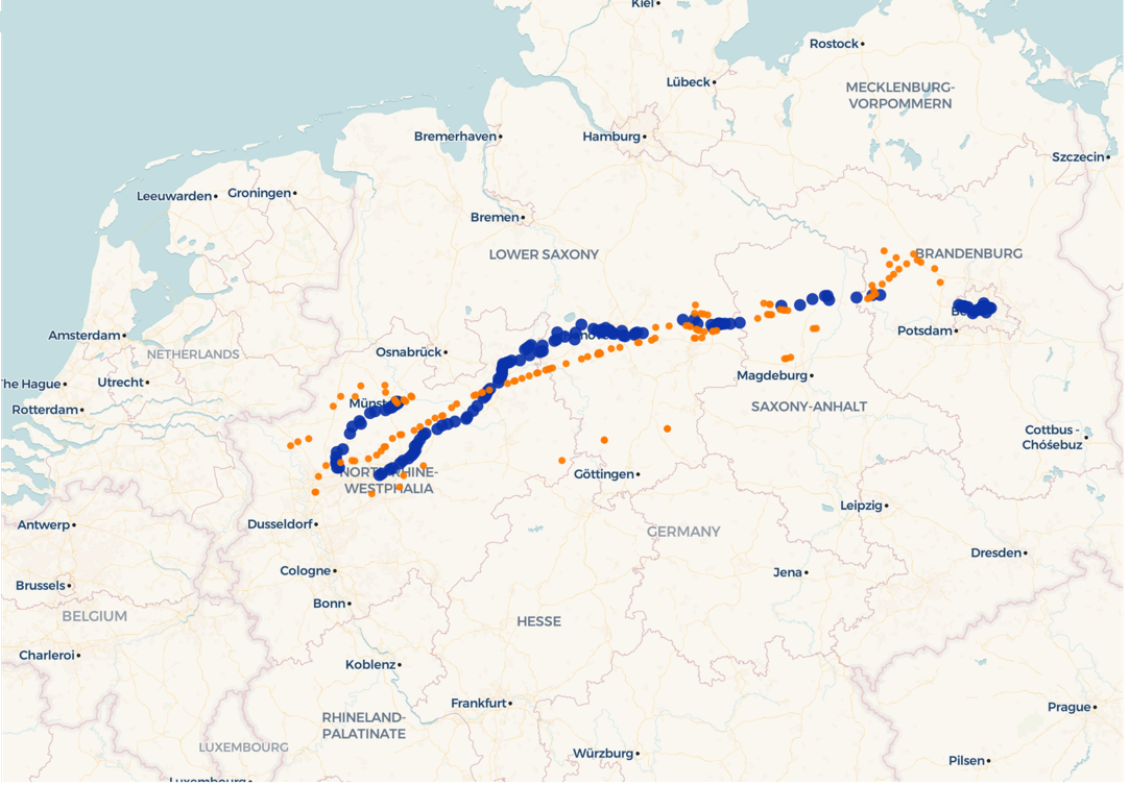}%
            \label{subfig:c}%
        }\hfill
        \subfloat[This Work, MSE=0]{%
            \includegraphics[width=.48\linewidth,trim=4cm 3cm 1cm 2cm,clip]{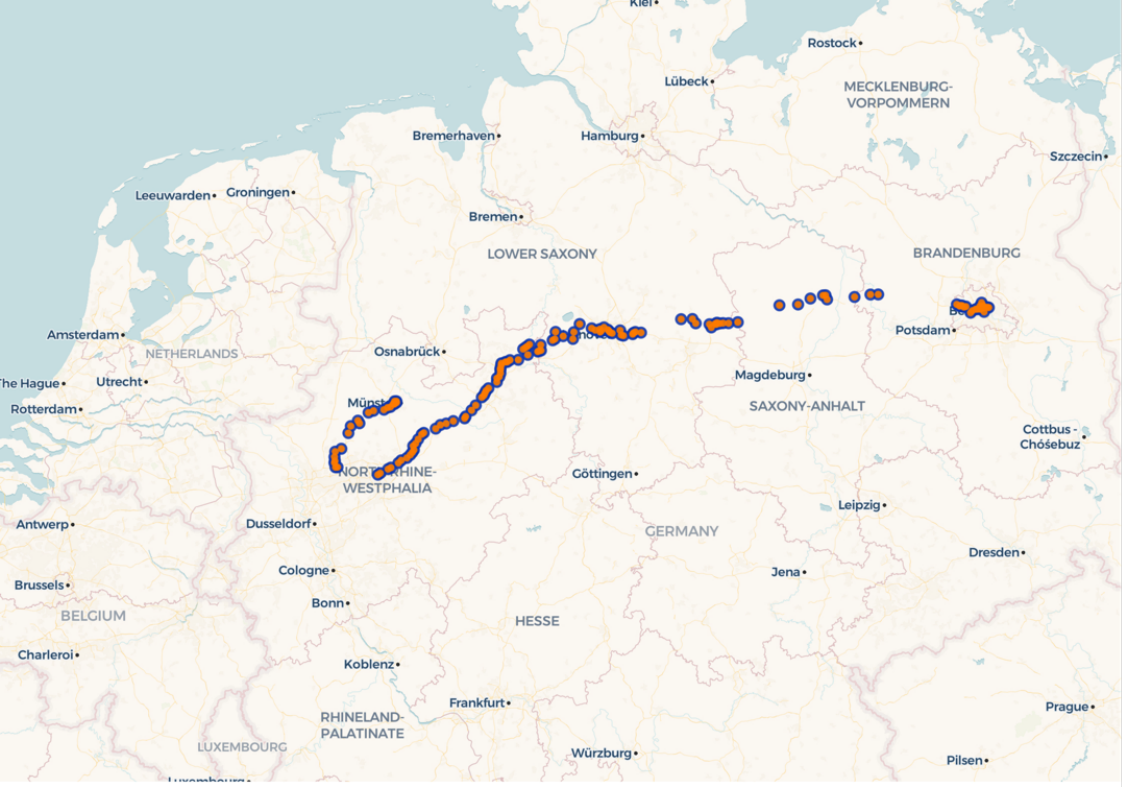}%
            \label{subfig:d}%
        }
        \caption{Comparison of state-of-the-art multidimensional leakage attacks~\cite{even_less,remin} and \framework on \texttt{Spitz} dataset. Plaintext values in blue, reconstructed in orange, under each attack's most favorable leakage. Query distribution is Gaussian.}
        \label{fig:SptizIntro}
    \end{figure}

\begin{table*}[t]
\centering
\caption{Comparison of leakage attacks on encrypted range queries. Formal guarantees focus on high-dimensional settings.}
\label{tab:comparison}
\resizebox{\textwidth}{!}{%
\begin{tabular}{cc l c cc cc cc c}
\toprule
& & & & \multicolumn{2}{c}{\textbf{Leakage}} & \multicolumn{2}{c}{\textbf{Assumptions}} & \multicolumn{2}{c}{\textbf{Formal Guarantees}} & \\
\cmidrule(lr){5-6} \cmidrule(lr){7-8} \cmidrule(lr){9-10}
& \textbf{Attacks} & \textbf{Query Type} & \textbf{All Databases} & \textbf{Search Pat.} & \textbf{Access Pat.} & \textbf{Query Dist.\ Known} & \textbf{Injected Plaintext} & \textbf{Observed Input} & \textbf{Output} & \textbf{Output Format} \\
\midrule
\midrule
\multirow{4}{*}{\rotatebox[origin=c]{90}{\textbf{1-D}}}
& Kellaris \emph{et al.}~\cite{DBLP:conf/ccs/KellarisKNO16}             & 1D Range & \cmark  & \xmark & \cmark & \cmark (Uniform)   & \xmark & --- & --- & Coordinates \\
& Lacharit\'e \emph{et al.}~\cite{DBLP:conf/sp/LachariteMP18}         & 1D Range & \xmark (Dense)   & \xmark & \cmark & \xmark  & \xmark & --- & --- & Coordinates \\
& Grubbs \emph{et al.}~\cite{10.1109/SP.2019.00030}                 & 1D Range & \cmark  & \xmark & \cmark & \cmark (Uniform)   & \xmark & --- & --- & Coordinates \\
& Kornaropoulos \emph{et al.}~\cite{state-of-the-uniform}   & 1D Range & \cmark  & \cmark & \cmark & \xmark  & \xmark & --- & --- & Coordinates \\
\midrule
\midrule
\multirow{3}{*}{\rotatebox[origin=c]{90}{\textbf{2-D}}}
& Falzon \emph{et al.}~\cite{fmacrst-fdrtd-20}                 & 2D Range & \cmark   & \cmark & \cmark & \xmark  & \xmark & --- & --- & Coordinates \\
& Falzon \emph{et al.}~\cite{fmacrst-fdrtd-20}                 & 2D Range & \cmark   & \xmark & \cmark & \cmark (Any)  & \xmark & --- & --- & Coordinates\\
& Markatou \emph{et al.}~\cite{10.1145/3460120.3484552}           & 2D Range & \cmark   & \cmark & \cmark & \xmark  & \xmark & --- & --- & Coordinates \\
\midrule
\midrule
\multirow{4}{*}{\rotatebox[origin=c]{90}{\textbf{High-D}}}
& Markatou \emph{et al.}~\cite{even_less} -- \texttt{Even Less}           & $d$D Range & \xmark (Dense) & \xmark & \cmark & \xmark  & \xmark & \xmark (Uniform QD) & \xmark & Topology \\
& Li \emph{et al.}~\cite{remin} -- \texttt{Remin}                               & $d$D Range & \cmark & \cmark & \cmark  & \xmark & \xmark & \xmark  & \xmark & Topology \\
& Li \emph{et al.}~\cite{remin}   -- \texttt{Remin+}                             & $d$D Range & \cmark & \cmark & \cmark  & \xmark & \cmark & \xmark  & \xmark & Coordinates \\
\cdashline{2-11}
& \textbf{This Work}  -- \framework                              & $d$D Range & \cmark & \xmark & \cmark  & \cmark (Any) & \xmark & \cmark & \cmark & Coordinates \\
\bottomrule
\bottomrule
\end{tabular}%
}
\end{table*}

\textbf{A History of Attacks on Encrypted Ranges.}
Leakage-abuse attacks on encrypted range queries have advanced significantly over the past 10 years~\cite{DBLP:conf/ccs/KellarisKNO16,DBLP:conf/sp/LachariteMP18,10.1109/SP.2019.00030,state-of-the-uniform,fmacrst-fdrtd-20,10.1145/3460120.3484552,even_less,remin}; Table~\ref{tab:comparison} provides a comprehensive overview. These attacks draw on a diverse set of techniques, spanning probabilistic analysis, statistical learning theory, computational geometry, algorithm design, and graph drawing. The evolution of this area can be traced through three distinct phases: the \emph{one-dimensional} era~\cite{DBLP:conf/ccs/KellarisKNO16,DBLP:conf/sp/LachariteMP18,10.1109/SP.2019.00030,state-of-the-uniform,response-hiding,cryptoeprint:2019:395}, which established the foundational attack paradigms; the \emph{two-dimensional} era~\cite{fmacrst-fdrtd-20,10.1145/3460120.3484552}, which revealed new geometric challenges; and the most realistic and demanding frontier, the \emph{multi-dimensional} era~\cite{DBLP:journals/popets/MarkatouFET23,even_less,remin}, where the curse of dimensionality makes reconstruction harder.

In this work, we tackle the most demanding setting of multi-dimensional encrypted range queries by revisiting one of the very first cryptanalytic approaches (used in~\cite{DBLP:conf/ccs/KellarisKNO16}, which started this research field), namely, \emph{knowledge of the query distribution}. 
Our attack is conceptually simple (compared to the state of the art~\cite{even_less,remin}) and, as demonstrated in Figure~\ref{fig:SptizIntro}, outperforms\footnote{For this comparison, we granted~\cite{even_less,remin} their most favorable observation scenario, namely access to all possible responses, while our attack's favorable scenario comprises accurate estimates of the retrieval frequency for each tuple of encrypted records.} the competition while providing rigorous guarantees (for the very first time) about (1) its desired number of observed queries, (2) its computation, and (3) the worst-case quality of its output.

\textbf{The First Attack with Query Distribution Knowledge.}
Several leakage-abuse attacks in the one- and two-dimensional case assume knowledge of the underlying query distribution~\cite{DBLP:conf/ccs/KellarisKNO16,10.1109/SP.2019.00030,fmacrst-fdrtd-20}. The core idea is to observe the retrieval rates of encrypted records and infer their plaintext values by comparing these rates to the expected retrieval frequencies induced by the known distribution.
Kellaris, Kollios, Nissim, and O'Neill introduced the first attack of this kind~\cite{DBLP:conf/ccs/KellarisKNO16}, establishing a paradigm that has influenced subsequent works in this area. Their adversarial strategy relies on \emph{frequency matching}: the attacker compares the empirically observed frequency with which each encrypted record is retrieved (derived from access-pattern leakage) against the theoretical probability that a given plaintext value would be retrieved (derived from the known query distribution). 
Inspired by this classical take, we push frequency matching to its limits in the multi-dimensional regime.

\textbf{Limitations of High-Dimensional Attacks.} There are several limitations with current high-dimensional attacks, i.e.,~\cite{even_less,remin}.

\emph{Output is a Topology, Not Reconstrcuted Coordinates.} 
Attacks \texttt{Even Less}~\cite{even_less} and \texttt{Remin}~\cite{remin} output a graph embedding as the main artifact of an attack. This graph embedding conveys only the topology (i.e., plaintext $r_1$ neighbors with plaintext $r_4$) without committing the reconstructed plaintexts to a concrete coordinate in the Euclidean space (see Column ``\emph{Output Format}'' in Table~\ref{tab:comparison}). 
It is worth noting that, despite the elegance of this graph-drawing approach, no prior leakage attack has focused on topology\footnote{The Approximate Order Reconstruction attack of~\cite{10.1109/SP.2019.00030} can be viewed as a form of topological reconstruction; however, being restricted to one dimension, the setting is considerably simpler than the multi-dimensional case.} as a reconstruction target in its own right; in part beccause topology has been regarded as \emph{an intermediate step} on the path to full coordinate reconstruction.

To bridge the gap between topological output and coordinate reconstruction, both~\cite{even_less,remin} apply post-attack transformations to the graph embedding, mapping it onto coordinates that can be directly compared against the original plaintext values. 
We stress that the attacker cannot perform these transformations in practice, since doing so requires access to the original plaintext coordinates, precisely the information the attack seeks to recover in the first place.
Specifically, the transformations are ($i$) a rotation, ($ii$) scaling, and ($iii$) a translation. 
To produce subfigures (b) and (c) in Figure~\ref{fig:SptizIntro}, we manually applied the combination of rotation, scaling, and translation that minimizes the Mean Squared Error (MSE) between the embedding's assigned coordinates and the original plaintext values. 

The issue with translating a topology to coordinates is that \emph{infinitely many coordinate assignments} are consistent with any given topology. 
That is, in a realistic scenario, for the attacker to output a reconstruction, they have to \emph{guess} a rotation, a scaling, and a translation, and hope that these guesses bring the resulting coordinates of the output embedding close enough to the original plaintext.

\emph{Unrealistic Assumptions.} Interestingly, Li et al.~\cite{remin} propose a variant called \texttt{Remin+} that sidesteps the transformation guessing problem described above, but this requires a much stronger assumption. 
Their approach requires the attacker to inject plaintext records with adversarially chosen values into the victim's database. 
These injected records serve as anchor points, allowing the attacker to pin the embedding to concrete coordinates without needing to guess the transformation. 
This is a significantly stronger assumption that limits the practical applicability of their approach.  

A subtle but critical issue with the injection technique proposed in~\cite{remin} is that the attacker must be able to identify \emph{which ciphertext} in the database corresponds to the injected plaintext; without this link, the anchor points cannot serve their purpose. 
Previous work on injection-based attacks for keyword search~\cite{DBLP:conf/uss/ZhangKP16} has argued that this identification is feasible in the context of keyword search through three mechanisms: (1) in a \emph{dynamic} searchable encryption scheme for keywords, the attacker can observe the (injected plaintext's)  ciphertext being created at the time of injection; (2) the attacker can force the victim to issue a keyword query that reveals which ciphertext corresponds to the injected record; and (3) the attacker can craft an injected file with a unique size, so that response volume alone reveals the anchor point's identity. 
However, \emph{none} of these arguments transfers to the setting of high-dimensional encrypted range queries. 
Mechanism~(1) fails because no efficient dynamic SE scheme exists for range queries in high dimensions. 
Mechanism~(2) is unrealistic, as it requires persuading the victim to issue a specific, adversarially chosen multi-dimensional range query. 
Mechanism~(3) does not apply because injected records are coordinate tuples with fixed size, not documents whose length the attacker can control.

The above discussion reveals a major gap in the literature: 

\begin{quote}
\emph{No existing cryptanalysis of high-dimensional encrypted range queries simultaneously produces coordinate-level reconstruction \underline{and} operates under a reasonable threat model, i.e., no chosen-plaintext injection into $\db$.}
\end{quote}

\noindent\textbf{Our Contributions.} We make the following contributions:

\begin{itemize}[leftmargin=*]
\item \textbf{Generalizing Frequency Matching with Guarantees.} Frequency matching has been used for one-dimensional range queries~\cite{DBLP:conf/ccs/KellarisKNO16,10.1109/SP.2019.00030}; we present its most general formulation as \framework (\underline{\bf L}eakage-\underline{\bf A}buse via \underline{\bf Ma}tching, see~\cref{sec:framework}), a three-component architecture for multi-dimensional databases in which the attacker observes access-pattern leakage and knows the query distribution (an assumption also used in prior cryptanalyses~\cite{DBLP:conf/ccs/KellarisKNO16,10.1109/SP.2019.00030,fmacrst-fdrtd-20}). We further establish rigorous guarantees (see~\cref{sec:guarantees}) on ($i$) the query complexity for an $(\epsilon,\delta)$-approximation of retrieval frequencies, ($ii$) an optimal parameterization of \framework that matches the information-theoretic reconstruction limit, and ($iii$) output quality reconstruction bounds relative to that limit---none of which are offered by any prior multi-dimensional attack on ranges~\cite{even_less,remin}.
\item \textbf{Query Distributions that Make Reconstruction Hard.} To understand the limits of this threat model, we study which query distributions maximize reconstruction uncertainty. 
We show that the gold standard, a distribution implying a reconstruction space spanning all possible databases, is unachievable. 
As a positive result, we construct a query distribution under which the reconstruction space contains every database that preserves the pairwise $L_1$ distances of $\db$, thereby providing the first purposefully structured and formally characterized reconstruction space.
\item \textbf{Evaluation.} We evaluate \framework on the same datasets, query distributions, and codebases as prior work~\cite{even_less,remin}, under both ideal and realistic observation scenarios. In the ideal case, \framework generates a very small set of solutions that contain the perfect reconstruction (MSE$=0$) across all datasets, while the state-of-the-art attacks~\cite{even_less,remin} generate tens of thousands of solutions that incur MSE orders of magnitude higher even under their most favorable leakage assumptions. In the realistic case, \framework's reconstruction error is smaller than~\cite{even_less,remin} and decreases monotonically with the number of observed queries, while the baselines show no such improvement.
\end{itemize}

\section{Background and Preliminaries}

\textbf{Notation.}
For any integer $y$, let $[y]$ denote the set $\{1,2,\ldots,y\}$. 
Let $[y]^k$ denote the $k$-fold Cartesian product, i.e., $\{1,2,\ldots,y\} \times  \ldots \times \{1,2,\ldots,y\}$ for $k$ sets. 
Let $a=(a_1,\ldots,a_k)$ and $b=(b_1,\ldots,b_k)$ be points in the $k$\emph{-dimensional plaintext domain} $\values=[N]^k$, and $a_i \leq b_i$ for all dimensions $i \in [k]$, then we say that $b$ \emph{dominates} $a$ (or equivalently, $a$ is dominated by $b$), denoted $a \preceq b$. 
To be consistent with previous works, we refer to points in $\values=[N]^k$ as values, even though they are $k$-dimensional vectors.
We define the \emph{distance} between two values $a,b \in \V$, denoted $dist(a,b)$, as the $L^1$ (or Manhattan distance):
    $dist(a,b)=\sum_{i=1}^k |a_i-b_i|.$

\textbf{Structured Encryption for Range Queries.}
Let $\values=[N]^k$ be the \emph{domain of values}, where $N$ and $k$ are positive integers. 
Let $\rids$ be the set of identifiers of the database used to uniquely identify encrypted records. 
A \emph{database} $\db=\{(\rid,v)\mid\rid \in \rids, v \in [N]^k\}$ is a collection of identifier-value pairs.
A dimension of $\values$ can be seen as a database attribute, e.g., ``\texttt{AGE}'' and each identifier as an encrypted medical file of a patient. 
We use \emph{record} and \emph{identifier} interchangeably, as each encrypted record has a unique identifier $\rid \in \rids$.
We denote the value $v$ of $\rid$ as $\db(\rid)$.

A \emph{structured encryption scheme for range queries (R-STE)} is a  primitive for encrypted search.
An R-STE scheme allows the client to encrypt and outsource $\db$ to a server and perform queries.
A range query $\query$ in  $\values$ can be seen as a \emph{hyperrectangle} in  $[N]^k$.
That is, a query $\query=[a,b]$ is defined by two vertices of the corresponding hyperrectangle, i.e., the vertex $a\in \values$ that is dominated by all other vertices of the hyperrectangle and the vertex $b\in \values$ that dominates all other vertices of the hyperrectangle. 
The universe of all queries is denoted as $\queries$.
We say that a query $\query=[a,b]$ \emph{covers} value $v$ if $a \preceq v \preceq b$.
In R-STE's query phase, the client issues an encrypted range query to the server.
The server responds with the identifiers whose values lie in the range specified by the query.
The set of returned identifiers is called a \emph{response}, denoted $\resp$, and the universe of responses $\responses$ is the power set of identifiers $\responses = \mathcal{P}(\rids)$.

\emph{Algorithms.} We define a static R-STE scheme consisting of the following algorithms: $\Setup$, which takes the security parameter $\lambda$ and $\db$ and outputs the secret key $\sk$ to the client and the encrypted database $\edb$ to the server; $\Trpdr$, which takes the secret key $\sk$ and the query $\query$ from the client and outputs a token (i.e., a trapdoor) for query $\query$ to the client; and $\Search$, which takes a token $t$ from the client and encrypted database $\edb$ from the server and outputs a set of identifiers $\resp \subseteq \rids$ to the client. 
More formally:

\begin{definition}
    Let $\Sigma=(\Setup,\Trpdr,\Search)$ be an R-STE scheme and let $\db$ be database over $\rids$ and $\values$.
    We say that $\Sigma$ is \emph{correct} if, for every $\query=(a,b)$ in $\queries$, after the execution of $(\sk,\edb)\leftarrow\Sigma.\Setup(\lambda,\db)$, $t\leftarrow\Sigma.\Trpdr(\sk,\query)$, and $\resp\leftarrow\Sigma.\Search(t,\edb)$, the following holds for $\resp$:
    $\resp = \{\rid \mid a \preceq \db(\rid) \preceq b \}$.
\end{definition}

For simplicity, we make one further modification: that every value $v \in \values$ is associated with at most one identifier $\rid \in \rids$. 
This is a standard simplification adopted for clarity of exposition. 
The extension to the general case, where multiple records may share the same value, is straightforward and well-understood; see the ``collocated-record identification procedure'' described in Algorithms~2 and~3 of~\cite{even_less}, which reduces the general case to the one-record-per-value setting as a preprocessing step.

\textbf{Leakage Profile.}
The information revealed to the server while running R-STE algorithms is defined as a set of functions over the plaintext data called \emph{leakage functions}.
Taken together, these functions make up the \emph{leakage profile} $\Lambda$ of a scheme, and are typically categorized as either \emph{setup leakage} $\lsetup$ or \emph{query leakage} $\lquery$, where $\Lambda=(\lsetup,\lquery)$.

Following the notation in~\cite{DBLP:conf/crypto/KamaraMO18}, we define three leakage functions relevant to our analysis.
The \emph{total response-length pattern} $\trlen$ takes $\db$ and outputs the total number of identifiers in all responses returned for queries $\query \in \queries$.
The \emph{response-identity pattern} $\ridp(\query)$ (often called the \emph{access pattern}) reveals, for each execution of $\Search$, the identifiers contained in the response.
The \emph{query-equality pattern} $\qeq$ (often called the \emph{search pattern}) takes an array of queries $[\query_1,\ldots,\query_M]$ and outputs an $M\times M$ binary matrix where $M[i,j]=1$ if $\query_i = \query_j$, and $M[i,j]=0$ otherwise.
A common leakage profile, both with respect to earlier constructions~\cite{DBLP:conf/esorics/FaberJKNRS15,Demertzis:2016:PPR:2882903.2882911} and cryptanalytic efforts~\cite{state-of-the-uniform}, for R-STE schemes is $\Lambda=\{\trlen,(\qeq,\ridp)\}$.

In this work, we focus on a less revealing leakage profile (i.e., a more challenging scenario for the attacker), $\Lambda=\{\trlen,\ridp\}$, to demonstrate the effectiveness of frequency matching even when the \emph{search pattern is suppressed}, as in ~\cite{DBLP:conf/crypto/KamaraMO18}.

\textbf{Query and Response Distributions.} 
Following Kellaris \emph{et al.}~\cite{DBLP:conf/ccs/KellarisKNO16}, we model client queries as i.i.d. samples from a distribution on the universe of queries $\queries$.
We call this distribution the \emph{query distribution}, denoted $\qdist$.
Formally, let $X$ be a random variable over $\queries$ that follows the distribution $\qdist$ and denotes a query issued by the client.
We denote the probability that $X = \query$ as $\Pr[X=\query] = \Pr_{\qdist}[\query]$.
We may omit the subscript $\qdist$ for brevity if it is clear from the context. 
Interestingly, by fixing a distribution on the universe of queries $\queries$, we fix a distribution on the universe of responses $\responses$ for a given $\db$, that we call the \emph{response distribution}, denoted $\rdist$.
Intuitively, the probability that a response $\resp$ will be returned is equal to the sum of probabilities of all queries that return exactly $\resp$. 
The response distribution is thus a function of both the database and the query distribution. 
Let $\queries(\resp,\db)$ denote the set of queries with response $\resp$ in database $\db$.  
Formally, let $Y$ be a random variable over the responses $\responses$ that follow $\rdist$.
We denote the probability that $Y=\resp$ as
\begin{equation}
\Pr[Y=\resp] = \Pr_{\rdist}[\resp]=\sum_{q \in \queries(\resp,\db)} \Pr_{q \sim \qdist}[q]
\label{eq:responses_sum}
\end{equation}
where again, we drop the subscript $\rdist$ when clear from context.
The last sum indicates that the probability that a response $\resp$ will be returned by a query sampled from $\qdist$ is the sum of the probabilities of all queries that return exactly $\resp$.

\textbf{Adversarial Goal.}
Let $\db$ be a database over the identifiers $\rids$ and domain $[N]^k$, and associated with query distribution $\qdist$, and let $\Sigma$ be an instance of the R-STE scheme defined above.
The adversary that we consider in this work is the server in $\Sigma$, who attempts to learn $\db$ by observing the access-pattern leakage.
As is common in the leakage cryptanalysis literature~\cite{DBLP:conf/ccs/KellarisKNO16,state-of-the-uniform,fmacrst-fdrtd-20}, we assume that the server knows the domain $[N]^k$ and universe of queries $\queries$.
Crucially, we assume that the adversary knows the query distribution $\qdist$. 
This assumption is consistent with the standard threat model adopted by prior attacks in one and two dimensions~\cite{DBLP:conf/ccs/KellarisKNO16,10.1109/SP.2019.00030,fmacrst-fdrtd-20}; we extend it here to the significantly more challenging multi-dimensional setting.

\section{LAM\MakeLowercase{a}: A FREQUENCY MATCHING ATTACK}

\label{sec:framework}

We introduce a framework called \underline{\bf L}eakage-\underline{\bf A}buse via \underline{\bf Ma}tching (\framework) for performing database reconstruction through frequency matching. At a high level, the attacker observes the \emph{empirical frequency} with which each subset of encrypted records is retrieved and attempts to \textbf{match it} against the \emph{retrieval probability} induced by each candidate plaintext assignment under the known query distribution $\qdist$. 
This attack strategy is conceptually simple and applies to an arbitrary number of dimensions. 
Yet despite being a folklore technique, it has never been studied or rigorously analyzed, leaving its effectiveness poorly understood.  
The absence of such a principled analysis is precisely what compelled prior work on multi-dimensional attacks~\cite{even_less,remin} to opt for more involved techniques, ultimately requiring stronger assumptions while, as our experiments demonstrate in~\cref{sec:evaluation}, giving worse reconstructions.

\textbf{On Frequency Matching.}
Since the adversary knows the query distribution $\qdist$, they can compute, for each plaintext value $v \in \values$, the probability that $v$ is covered by a query sampled from $\qdist$:
\begin{displaymath}
\Pr[v] = \sum_{q\, :\, q \text{ covers } v} \Pr_{q \sim \qdist}[q].
\end{displaymath}
A key observation is that this \emph{retrieval probability} can be \textbf{precomputed} entirely offline, without observing a single query, since it depends solely on the knowledge of $\qdist$. 
We now turn to the attacker's perspective after observing query leakage. Suppose the adversary observes a multiset $M$ of $|M|$ responses to queries sampled from $\qdist$. We define the \emph{empirical frequency} of a record identifier $\rid$ as the fraction of observed responses that contain $\rid$:
\begin{equation}
\label{eq:freq}
f(\rid) := \frac{|\{\resp \in M\, :\, \rid \in \resp\}|}{|M|}.
\end{equation}
In the limit, as $|M| \to \infty$, the empirical frequency $f(\rid)$ converges to the retrieval probability of the true plaintext value of $\rid$. 
Thus, for a fixed record $\rid$, the attacker can test all candidate values $v \in \values$ against the precomputed probabilities and identify those for which $f(\rid) \stackrel{?}{=} \Pr[v]$. 
Each match yields a candidate plaintext assignment $\db(\rid) = v$ that is consistent with the observed leakage.

In the equality test above, the attacker compares the empirical frequency of a \emph{single} record to the probability of a \emph{single} value.
However, the attacker can devise more \emph{fine-grained} comparisons by using \emph{set operations} over a subset of encrypted record retrievals.
E.g., the attacker can compare the frequency with which two encrypted records are simultaneously retrieved (as part of a response) to the probability of two values being simultaneously covered by a query: 
\begin{equation}
f(\rid \cap \rid') \stackrel{?}{=} \Pr[v \cap v'],
\label{eq:equality}
\end{equation}
which essentially tests whether $\rid=v$ \emph{and simultaneously} $\rid'=v'$ (in fact, this equality is true even when $\rid=v'$ and $\rid'=v$). 
However, a match between the empirical frequency and a precomputed probability does not (necessarily) uniquely identify the true plaintext values. Multiple distinct assignments of values may satisfy this subset\footnote{We abuse the term ``subset'' to refer to tests that study simultaneous retrieval of encrypted records as part of a response; a subset can also be seen as an intersection of probability events.} equality, and more equality tests are required to be resolved.


\textbf{Overview of \framework.}
\framework consists of the components $\selector$, $\transl$, and $\solver$, where each one has a  synergistic input and output. 
This modular architecture admits multiple instantiations. At a high level, the $\selector$ identifies a collection of subsets of encrypted records whose retrieval frequencies will serve as the basis for the attack (like Equation~(\ref{eq:equality})). 
The $\transl$ then performs two tasks: ($i$) for each subset chosen by the $\selector$, it identifies all candidate plaintext value assignments whose precomputed retrieval probability under $\qdist$ matches the observed empirical frequency, and ($ii$) it encodes these frequency-probability matching pairs as a constraint programming formula. 
The $\solver$ takes this formula, finds a satisfying assignment of values, and outputs a database that is consistent with all subset equality tests simultaneously.

\textbf{$\selector:$ Choosing Record Subsets.}
The $\selector$ component is responsible for determining which record retrieval subset expressions (e.g., $\rid_c, \rid_a\cap\rid_c, \rid_c\cap\rid_f\cap\rid_z,\ldots$) will be used for the attack.
$\selector$ has a large impact on the runtime and accuracy of the attack.
If the number of subsets it chooses is too few, there might be too many reconstructions that match the frequency of the chosen subsets.
If the number of subsets it chooses is too large, the runtime may increase dramatically.
\begin{figure}[h]
    \centering
    \includegraphics[scale=0.45]{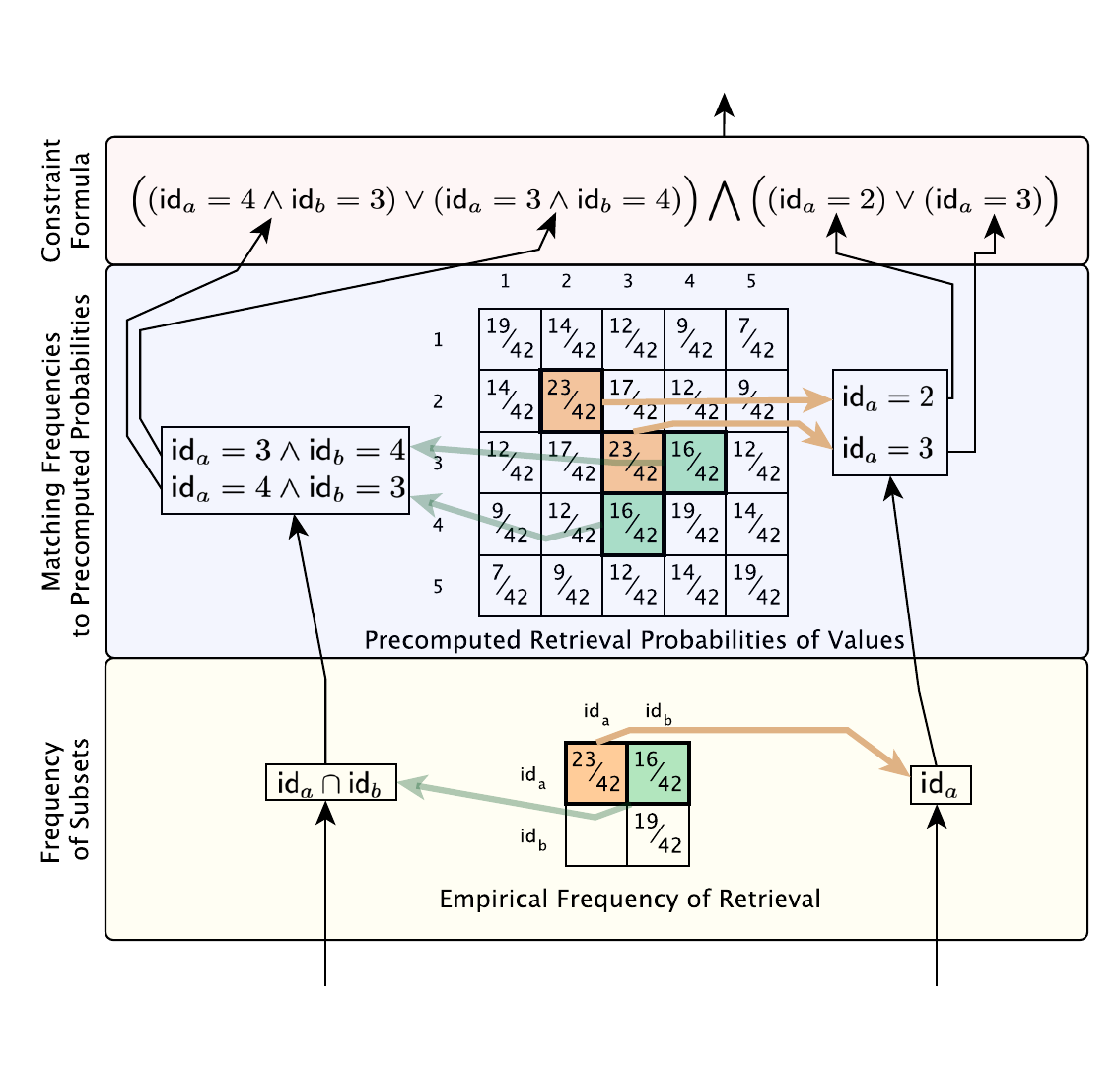}
    \caption{Internal view of $\transl$. Takes as input subsets to match $\{\rid_a\cap\rid_b,\rid_a\}$ from $\selector$. $\transl$ identifies frequency-matching events and generates the constraint programming formula passed to the $\solver$.
    }
    \label{fig:translator_flow}
\end{figure}

\textbf{$\transl$: One Constraint Formula from All Matchings.}
The $\transl$ component serves as the bridge between frequency matching and constraint programming, i.e., CP-SAT. 
Its primary role is to identify all plaintexts that match the empirical frequency. 
For any fixed encrypted record subset, there may be several plaintext assignments that satisfy this equality, meaning that multiple candidate plaintext assignments must all be considered. 
Once all plaintext assignments have been identified, $\transl$ converts them into a single logical formula $C$ that systematically constrains the space of possible plaintext-to-record assignments before handing $C$ off to the solver. 
The logical structure of $C$ is based on two observations about how matching subsets relate to one another. 
First, when multiple $C$ components share the same encrypted record subset (e.g., multiple candidate plaintexts for the same encrypted records $\rid_a$), the assignments they imply are mutually exclusive, since each record (treated as a variable of $C$) can hold at most one value; such alternatives are therefore disjoined, i.e., use of $\lor$. 
For example, in Figure~\ref{fig:translator_flow} the two alternative plaintext assignments $(\rid_a=4\land\rid_b=3)$ and $(\rid_a=3\land\rid_b=4)$ are ``glued'' with $\lor$.  
Second, since the true underlying database must be consistent with every subset expression simultaneously, the constraints produced by different subsets (e.g., subset $\rid_a\cap\rid_b$ and subset $\rid_a$) cannot be considered in isolation and are all required to hold simultaneously, i.e., use of $\land$. 
For example, in Figure~\ref{fig:translator_flow} the two subset expressions $(\rid_a=4\land\rid_b=3)\lor(\rid_a=3\land\rid_b=4)$ and $(\rid_a=2\lor\rid_a=3)$ are ``glued'' with $\land$. 


\textbf{$\solver:$ Reconstruction as Constraint Programming.}
$\solver$ takes the formula $C$ output by $\transl$ and solves the constraint programming problem to find an assignment of values to encrypted records that satisfies $C$.
Notice that any assignment that satisfies $C$ will have exactly one value per record.
Furthermore, any assignment that satisfies $C$ will satisfy all subsets output by $\selector$, and one of these possible assignments is guaranteed to be the true database $\db$, in the limit.
In our example, there is only one assignment that satisfies the constraints, that is $\rid_a=3$ and $\rid_b=4$. 
To recover multiple candidate assignments, $\solver$ can be run iteratively: each time a satisfying assignment is found, a new constraint is added to $C$ that explicitly excludes it, forcing $\solver$ to search for a different solution in the next iteration. This process can continue until no further satisfying assignments exist, at which point the full set of $C$-satisfying assignments has been enumerated. 
Since, in the limit, the true database $\db$ is guaranteed to satisfy $C$ by construction, it is certain to appear among them.

\textbf{Practical Challenges in \framework.} In what follows, we identify three fundamental challenges that must be addressed to establish \emph{rigorous guarantees} for a practical frequency-matching attack:
\begin{enumerate}
\item \textbf{Choice of Subsets to Match:} Which subsets of encrypted records should \framework target for matching? Since the number of possible subsets is exponential, exhaustive enumeration is infeasible, and an insightful selection criterion is needed.
\item \textbf{Query Complexity:} How many queries must the attacker observe before the \framework attack becomes effective?
\item \textbf{Output Quality:} What formal guarantees can be provided about the accuracy of the reconstruction?
\end{enumerate}

Regarding challenge (1), Figure~\ref{fig:translator_flow} illustrates the simple case of two subsets, $\rid_a\cap\rid_b$ and $\rid_a$, but in general, it is unclear how an attacker should select the right combination of subsets in order to properly constrain the space of reconstructions and avoid outputs that do not agree with the observed lekage; this is addressed in Section~\ref{sec:its_limit}. 
Regarding challenge (2), the equality in equation (\ref{eq:equality}) holds only asymptotically, making \emph{exact} matches hard to guarantee in practice. 
The first challenge \framework must address (see Section~\ref{sec:querycomplexity}) is therefore determining how many queries suffice for an ($\epsilon, \delta$)-approximation of the retrieval frequencies, relaxing exact equality to a within-$\epsilon$ match. 
Finally, challenge (3) calls for a principled technique that provides worst-case guarantees on the closeness of the reconstruction to the true plaintext; see Section~\ref{sec:output}.

\section{Provable Guarantees for LAM\MakeLowercase{a}}
\label{sec:guarantees}

\subsection{Choosing Subsets to Match, with Guarantees}
\label{sec:its_limit}

Before presenting a universal $\selector$ strategy with guarantees, it is useful to characterize the fundamental limits of \emph{any attack} under our threat model.
Although the attacker's goal is to recover the true database $\db$, perfect recovery is not always achievable. Under the leakage profile $\Lambda=\{\trlen,\ridp\}$, which excludes the search pattern $\qeq$, the only signal available for inference is the sequence of observed encrypted responses, drawn i.i.d. from the response distribution $\rdist$. 
When two distinct databases $\db$ and $\db'$ induce identical response distributions under a fixed query distribution $\qdist$, no attack can distinguish~\footnote{To see this, observe that given fixed $\db$ and $\qdist$, if there exists a $\db'$ that gives the same $\rdist$ as the one induced by ($\db,\qdist$); for every query sequence in $(\db,\qdist)$ there exists an equally probable query sequence in $(\db',\qdist)$ with identical responses.} between them using only the information available under our threat model. 
To formalize this inherent limit, we adapt the notion of a \emph{reconstruction space} (originally introduced by Kornaropoulos \emph{et al.}~\cite{k-nn-attack}) to our setting:
\begin{definition}[Reconstruction Space]
    Let $\db$ be a database, let $\qdist$ be a query distribution over $\db$, and let $\rdist$ denote the response distribution induced by $\qdist$ on $\db$. The \emph{reconstruction space} $\rs{\qdist}(\db)$ is the set of all databases whose response distribution under $\qdist$ equals $\rdist$.
    \label{def:RS}
\end{definition}
In this setting, the best any attack can hope to achieve is to output only databases within $\rs{\qdist}(\db)$. Recovering exactly the databases in $\rs{\qdist}(\db)$ is therefore the best any attack can achieve under this threat model, and we prove that \framework meets this information-theoretic limit. 
Our reconstruction space differs from~\cite{even_less} (no query distribution knowledge) and~\cite{k-nn-attack} (specific to $k$-NN queries).

\textbf{The Proposed $\selector$: $T_{2k}$.}
We present a $\selector$ strategy that simultaneously satisfies three properties: ($i$) it is \emph{optimal}, in the sense that, if a database $\db'$ satisfies the constraints by the proposed $\selector$, then $\db' \in \rs{\qdist}(\db)$, thereby matching the information-theoretic limit; ($ii$) it is \emph{compact}, producing a number of subsets that scales $O(n^{2k})$ with the number of dimensions $k$ and the size $n$ of the $\db$; and ($iii$) it is \emph{universal}, meaning these guarantees hold for every database and every query distribution. 

Informally, the collection $T_{2k}$ of the proposed $\selector$ contains the following subsets: for every $j \in [2k]$, all $\binom{n}{j}$ possible intersections of $j$ records from $\db$. Listing these subsets explicitly, for $j=1$ we get the singletons $\rid_1,\ldots,\rid_n$; for $j=2$ the pairwise intersections $\rid_1\cap \rid_2,\ \rid_1\cap \rid_3,\ \ldots,\ \rid_{n-1}\cap \rid_n$; for $j=3$ the triple intersections $\rid_1\cap \rid_2\cap \rid_3,\ \ldots$; and so on up to the $2k$-fold intersections. More formally, for a set of records $\rids$ and domain $\values=[N]^k$, the collection $T_{2k}$ is defined as:
\begin{displaymath}
T_{2k}=\Big\{\rid_{i_1},\ \rid_{i_1}\cap\rid_{i_2},\ \ldots,\ \rid_{i_1}\cap\cdots\cap \rid_{i_{2k}}\ \Big|\ 1\leq i_1 < \cdots < i_{2k} \leq n \Big\}
\end{displaymath}

\noindent Notice that the number of subsets in $T_{2k}$ are $\sum_{j=1}^{2k}\binom{n}{j}=O(n^{2k})$.


\textbf{$T_{2k}$-Reconstructions Come From the Reconstruction Space.}
In the following, we show that \emph{any} $\db'$ that satisfies the constraint programming instance implied by $T_{2k}$, has the same response distribution as $\db$ under $\qdist$. 
Since the attacker in this threat model cannot prioritize over members from $\rs{\qdist}(\db)$, this means that the proposed attack is optimal, in the limit. 
As a first step, we give a lemma concerning the geometry of queries in $\values=[N]^k$.

\begin{lemma}
\label{lemma:covering}
    For any set of plaintext values $V \subseteq \values=[N]^k$, there exists a subset $V^* \subseteq V$ of size at most $2k$ such that any query covering all values in $V^*$ also covers all values in $V$.
\end{lemma}

Figure \ref{fig:minmax} illustrates Lemma~\ref{lemma:covering} in the context of record retrieval: in any response $\resp$ with $V$ associated values, there is always a set $V^*$ of $2k$ or fewer values from $V$ such that every query covering $V^*$ also covers all other records in $V$.

 \begin{figure}[h]
\centering
\includegraphics[scale=0.32]{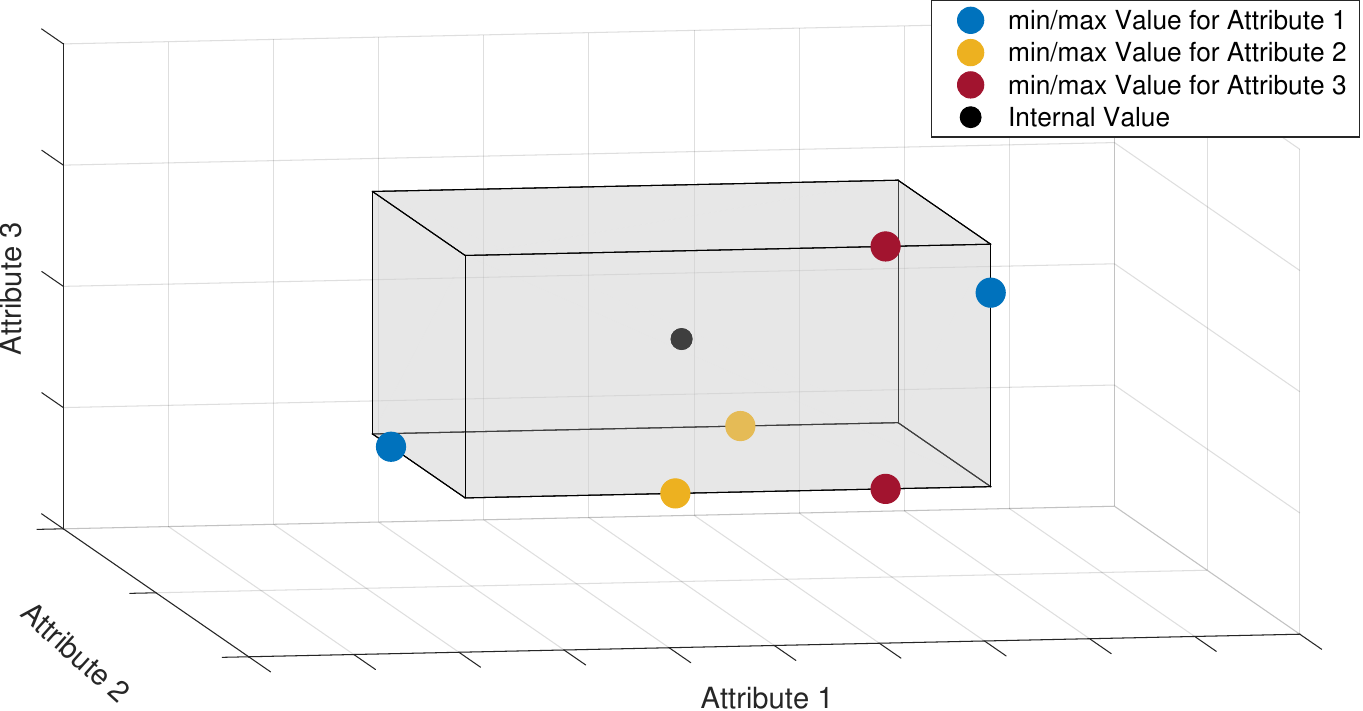}
\caption{Values of a $\db$ in $k$=$3$ dimensions.  Red/blue/yellow/ black points comprise the associated values of records that are part of $\resp$. Queries that cover the ``colored'' values from $\resp$ \emph{must} cover all records of $\resp$ with internal values. 
}
\label{fig:minmax}
\end{figure}

\noindent The above Lemma is used in the proof of the next theorem.

\begin{theorem}
Let $\db$ and $\db'$ be databases over $\rids$ and $\values=[N]^k$, and let $\qdist$ be a query distribution over both. Let $\rdist$ and $\rdist'$ be the response distributions of $\db$ and $\db'$ respectively, both induced by $\qdist$. If $\db'$ satisfies the \framework $T_{2k}$-constraint programming formula then $\rdist = \rdist'$, in the limit.  
    \label{theorem:2ktuples}
\end{theorem}

\textbf{Subsets Smaller Than $T_{2k}$ Give Suboptimal Results.}
Theorem~\ref{theorem:2ktuples} establishes $2k$ as an upper bound on the size of intersection expressions needed for \framework to output only databases in $\rs{\qdist}(\db)$. We now show that this bound is tight. Concretely, we show that the parameterization $T_{2k-1}$, which is identical to $T_{2k}$ but excludes subset expressions of size exactly $2k$, gives reconstructions outside $\rs{\qdist}(\db)$. 
We show that for any dimension $k$, there exists a database $\db' \neq \db$ that satisfies the constraint programming induced by $T_{2k-1}$ but not the constraint programming induced by $T_{2k}$; thus, confirming that expressions of size $2k$ are strictly necessary.

\begin{theorem}
\label{theorem:tightbound}
    For any dimension $k$ and $N \geq 6$, there exist a domain $\values=[N]^k$, databases $\db \neq \db'$ over $\rids$, and a query distribution $\qdist$ such that $\db'$ satisfies the $T_{2k-1}$-constraints of \framework but not the $T_{2k}$-constraints, in the limit.
\end{theorem}

Notice that the choice of $T_{2k}$ as the $\selector$ strategy is optimal in the limit \emph{without using the knowledge of $\qdist$ to choose subsets}, making the attack universal, i.e., it achieves optimal reconstruction guarantees regardless of $\qdist$ and $\db$. While a more succinct selector strategy may exist when $\qdist$ is part of the $\selector$ decision-making, such a strategy would only improve \framework's efficiency\footnote{E.g.,~\cite{DBLP:conf/ccs/KellarisKNO16} tailors frequency matching to uniform $\qdist$ (as opposed to our universal approach), reducing the $O(n^2)$ candidate pair subsets of $\selector$ to just $n$ pairs.}, not reconstruction quality, since $T_{2k}$ is already optimal in the limit.


\subsection{Query Complexity, with Guarantees}
\label{sec:querycomplexity}

The discussion so far has centered on the ideal (for the adversary) case, the so-called ``in the limit'', in which the frequency of observed subsets of encrypted records \textit{matched exactly} the true probability of that subset. 
However, in the real world, this is highly unlikely. An important question is then: Given some number of observed responses, how can we determine which true probabilities ``match''? 

A naive approach would be to simply calculate the (partially inaccurate) frequency and choose the plaintexts closest match. 
This may result in a constraint programming problem with \emph{no feasible solution} since we have no sense of how close the calculated retrieval frequency is to the true probability. 
Instead, we a statistical learning theory tool known as the \emph{Vapnik–Chervonenkis dimension} (VC-dimension) to get approximation bounds.

\begin{remark}
Given that the attacker can exactly learn $\rdist$ only in the limit, a realistic approach is to approximate $\rdist$ by the empirical frequency of subsets in a fixed sample from $\rdist$.
\end{remark}

\noindent We use VC-dimension to bound the number of responses required for the empirical approximation of $\rdist$ to be $\epsilon$-close to the true $\rdist$ with probability at least $1-\delta$, for any desired accuracy $\epsilon \in (0,1)$ and confidence $1-\delta \in (0,1)$.  

 \textbf{Approximating via Sampling.} 
Recall that $\rids$ is the universe of encrypted records and $\rdist$ is the distribution over the set $\responses = 2^{\rids}$. 
Let $\responses_{\db}\subseteq \responses$ be the set of all possible responses that can be returned by issuing queries from $\queries$ on $\db$, i.e., all possible responses for $\db$.
A multiset of observed responses $\multiset$ is a collection of \textit{independent identically distributed} (i.i.d.) samples from $\rdist$. 
We can calculate the frequency of some $\rid \in \rids$ by counting all responses in $\multiset$ that contain $\rid$, as in Equation~(\ref{eq:freq}). 
To find the frequency, we first need to compute the numerator in Equation~(\ref{eq:freq}), that is, the multiset of responses from $M$ each containing $\rid$. 
Ideally, we want to find the frequency of a subset $I=\{\rid_1, \rid_2, \ldots, \rid_n\}$ in $M$, $f(\rid_1 \cap \rid_2 \ldots \cap \rid_n)$ or more succinctly $f(I)$. 
To compute the above frequency of  $I$, we need to identify the multiset of $\multiset$, denoted as $\support_{\multiset}(I)$, where each member $\resp$  contains all $I$, more formally $\support_{\multiset}(I) = \{ \resp \in \multiset : I \subseteq \resp \}$.  
The frequency is then $f(I) = \frac{|\support_{\multiset}(I)|}{|\multiset|}$. 
Since $f(I)$ is an empirical approximation, an exact match between $f(I)$ and the true probability is unlikely. Instead, we relax the exact equality of Equation~(\ref{eq:equality}) to an \emph{approximate criterion}, that is, we identify an \emph{interval of frequencies} and seek matches across all of them.
This interval is formalized as, given $M$ and $I\in \responses$, find $\epsilon \in (0,1)$ and $\delta \in (0,1)$ such that with probability at least $1-\delta$ we have:
\begin{equation}
    \left| \left( \sum_{\resp \in \responses : I \subseteq \resp}\Pr_{\rdist}[\resp]\right) - \frac{\left|\support_{\multiset}(I)\right|}{\left|\multiset \right| } \right| \leq \epsilon.
    \label{eq:eps_bound}
\end{equation}

\textbf{Defining the Range Space.} We now define the fundamental notions required to derive VC-dimension bounds, adopting statistical learning theory terminology: 
The \emph{domain} is $\responses_{\db}$, the set of all possible responses for $\db$. 
For any subset of records $I \subseteq \rids$, we define our \emph{range set}\footnote{The term ``range'' is overloaded; in most of this work it refers to encrypted range queries, whereas in VC-dimension theory it denotes the subsets of the domain used to measure combinatorial complexity.} as the collection of all responses that contain $I$,  for every $I \subseteq \rids$ the range set contains $T_{\responses}(I)$.

\begin{definition}
Let $\responses_{\db}$ be the set of possible responses for a fixed $\db$. 
For each non-empty $I \subseteq \rids$, define $\support_{\responses_{\db}}(I) = \{\resp \in \responses_{\db} \mid I \subseteq \resp\}$ as the set of all responses of $\db$ containing $I$. The family $\texttt{R} = \{\support_{\responses_{\db}}(I) \mid I \subseteq \rids, I \neq \emptyset\}$ together with the domain $\responses_{\db}$ forms a \textbf{range space} $(\responses_{\db}, \texttt{R})$.
\label{def:range_pair}
\end{definition}

\textbf{Bounding VC-dimension. } The last remaining piece for deriving the query complexity is to define what it means to \textit{project} $\texttt{R}$ onto a collection of responses $S \subseteq \responses$. 
The projection of $\texttt{R}$ on $S$ is the set $\texttt{R}_S = \{ S \cap X \; | \; X \in \texttt{R} \}$. 
The VC-dimension $\mathrm{VC}(\texttt{R})$ is the cardinality of the largest collection $S \subseteq \responses_{\db}$ such that the projection of $\texttt{R}$ on $S$ is the powerset of $S$ (referred to as \textit{shattering}). 

Inspired by~\cite{Matteo_2014_true_itemsets}, we derive a bound on the VC-dimension of the range space $(\responses_{\db}, \texttt{R})$ by finding the optimal solution of the \textit{Set-Union Knapsack Problem}. For the definition of the Set-Union Knapsack Problem and the proof of the Theorem, see Appendix \ref{appendix:proofs}.

\begin{theorem}
    Let $\multiset$ be a sample drawn from $\rdist$.  
    Let $\ell = |\{\rid \in \rids \mid \exists  S \in \multiset \text{ s.t. } \rid \in S\}|$ be the number of encrypted records that appear in at least one response from $\multiset$. 
    Let $q$ be the optimal profit of the Set-Union Knapsack Problem over $\multiset$ with capacity equal to $\ell$. Let $b = \lfloor \log_2 q \rfloor + 1$. Then $\mathrm{VC}(\texttt{R}) \le b$.
    \label{thm:vc}
\end{theorem}

\noindent Equation~(\ref{eq:eps_bound}) defines the $\epsilon$-approximation to $(\texttt{R}, \rdist)$. Combined with Theorem~\ref{thm:vc}, this yields a rigorous sample complexity bound for the frequency estimates used in $\framework$:

\begin{theorem}[Thm. 2.12 \cite{Har-Peled2011_vc_dim_samples}]
    Let $\texttt{R}$ be a range set on $\responses_{\db}$ with $\mathrm{VC}(\texttt{R}) \le d$, and let $\rdist$ be a distribution on $\responses_{\db}$. Given $\delta \in (0, 1)$ and a positive integer $\ell$, let
\begin{equation}
\varepsilon = \sqrt{\frac{c}{\ell} \left( d + \log \frac{1}{\delta} \right)}
\end{equation}
where $c$ is an universal positive constant\footnote{\cite{Loffler2009ShapeFitting} experimentally demonstrated that c is, at most, 0.5}. Then, a multiset of $\ell$ elements of $\responses_{\db}$ sampled independently according to $\rdist$ is an $\varepsilon$-approximation to to range set $\texttt{R}$ and distribution $\rdist$ with probability at least $1 - \delta$.
\end{theorem}

\subsection{Output Quality, with Guarantees}
\label{sec:output}

Since the reconstruction space may contain multiple equally plausible plaintext assignments, a natural question arises: which candidate should the attacker commit to as their coordinate-level output? 
In this work, we propose two generic output strategies (applicable not only to \framework but also to prior attacks~\cite{even_less,remin}) for selecting \emph{a single database} from the reconstruction space.

\textbf{Diameters and Centroids.} In the following, we assume that the attacker derived $K$ plausible plaintext databases from the reconstruction space, i.e., $\{\db_1,\ldots,\db_K\}$. 
As a next step, for each encrypted record $\rid_j$, the attacker considers the $K$ candidate plaintext values assigned to it across these databases, namely $\{v^{(j)}_1,\ldots,v^{(j)}_K\}$.

For the first output strategy called \texttt{diameter}, the attacker computes the convex hull of $\{v^{(j)}_1,\ldots,v^{(j)}_K\}$ for each encrypted record $\rid_j$. 
The attacker then identifies the \emph{diameter} of the convex hull above, with length $\gamma_j$, and outputs the diametral midpoint as the reconstructed plaintext for $\rid_j$. 
If the true plaintext lies within the convex hull, the reconstruction error for $\rid_j$ is at most $\gamma_j/2$. 
The overall output quality guarantee\footnote{We note that the next $\epsilon$ concerns the approximation of the reconstruction error, while the $\epsilon$ in Equation (\ref{eq:eps_bound}) has to do with the approximation of the probability of retrieval via the empirical frequency.} across all records is then 
$$\epsilon := \max_{j\in [K]} \frac{\gamma_j}{2}.$$
The above strategy provides worst-case approximation guarantees but adopts a pessimistic view of the reconstruction.
For the second output strategy called \texttt{centroid}, the attacker computes the \emph{centroid} (or barycenter) of $\{v^{(j)}_1,\ldots,v^{(j)}_K\}$ for each encrypted record $\rid_j$. 
Geometrically, the centroid is the point that minimizes the sum of squared distances to all candidate plaintexts, and can be interpreted as the ``average'' location of the reconstruction space for $\rid_j$. 
When the true plaintext is not an outlier within the reconstruction space, the centroid tends to be closer to it than the diametral midpoint, making it a more optimistic approach.

The analysis of reconstruction space properties (such as its geometry and entropy) was introduced by Kornaropoulos et al.~\cite{DBLP:conf/ccs/KornaropoulosMP22}, establishing a line of work that has been extended by subsequent studies~\cite{DBLP:journals/popets/BoldyrevaGW24,DBLP:journals/cic/EspirituKM25}. 
It should be noted that a related diameter-based approach was proposed for encrypted $k$-NN queries~\cite{k-nn-attack}; however, in that setting, each convex hull vertex represents an entire database, whereas here the convex hull is taken over the plausible coordinate assignments of a single encrypted record.

\section{Which Query Distributions Make Reconstruction Hard?}
\label{sec:flattening}

Having shown that, in the limit, a database $\db$ with query distribution $\qdist$ can be reconstructed up to $\rs{\qdist}(\db)$, we ask: 
\begin{center}
``\emph{Is it possible to increase the reconstruction space, \rs\qdist(\db), and thus make reconstruction harder by increasing the uncertainty?}''
\end{center}
To achieve this increase, one has to change $\db$, $\qdist$, or both.
Altering the data is not ideal, as it may undermine the correctness of  \emph{R-STE}.
To find the hardest-to-reconstruct case, we pivot to identifying an ``uncertainty-increasing'' query distribution $\qdist$.

 One approach would be to tailor the query distribution to the underlying $\db$, but this strategy \emph{directly leaks} information about $\db$ since we assume that $\qdist$ is known to the attacker. 
Thus, we only study query distributions that are \emph{independent of the database $\db$ they operate on}.
This way, our analysis holds regardless of which database is queried.
In particular, we focus on \emph{expressive query distributions}, i.e., distributions where every query $\query \in \queries$ has a non-zero probability of being issued.
We avoid non-expressive distributions because forbidding queries degrades the scheme's functionality.

As an affirmative answer to the above question, we show an expressive distribution $\widehat{\qdist}$ which guarantees, for any database $\db$, a corresponding $\rs{\widehat{\qdist}}(\db)$ containing all databases with the \emph{same pairwise $L^1$ distances} as those of $\db$.  
Thus, even if the attacker successfully runs \framework against $\widehat{\qdist}$, they are met with a large $\rs{\widehat{\qdist}}(\db)$ that, for the first time in this literature, carries a \emph{purposefully designed structured reconstruction space} (unlike the uniform distribution, with a reconstruction that admits reflections across axes incidentally): the attacker can infer the pairwise $L^1$ distances of $\db$ but nothing beyond them.

\subsection{On the Impossibility of the Hardest Query Distribution}
\label{singleton_flattening}

The gold standard would be a query distribution $\qdist^*$ that implies an $\rs{\qdist^*}(\db)$ that consists of every possible database over $\rids, \V$.
In such a case, an attacker would have no advantage over a random guess from the set of all possible databases over $\rids, \V$. 
In the following, we show that such a query distribution \emph{can not exist}.

\textbf{A First Step Towards the Gold Standard $\qdist^*$.} Recall from Section~\ref{sec:its_limit} that any two databases that belong to the reconstruction space must have the same response distribution.
Thus, to identify the $\qdist^*$, we need to find a query distribution for which every possible plaintext value is queried (i.e., returned as part of a response) with the exact same probability, say $p^*_1$. 
Similarly, every possible pair of plaintext values is queried with the exact same probability, say $p^*_2$. 
Under such $\qdist^*$ \emph{every possible database} would have the same response distribution, so the attacker cannot infer anything about the underlying plaintext values.

As a first step towards this goal, we show in Algorithm \ref{algo:singleton} (in Appendix~\ref{sec:appendix_algo}
) how this can be done for $1$-tuples of records by only adjusting the probabilities of 1-tuples of values in any input distribution $\qdist$.
Our approach will impose a ``minimal'' change to the input $\qdist$ by increasing the probabilities of only the smallest queries (those that cover a single value). 
For simplicity of exposition, we assume that every query $\query$ is associated with a weight $w_{\query}$ (which is a natural number) and that the probability of this query $\query$ is given by normalizing its weight divided by the sum of all query weights. 
Algorithm~\ref{algo:singleton} simply finds the value $v \in \V$ with the highest retrieval probability $p^*_1$, then raises the weight of each singleton query $[v',v']$ until every value matches $p^*_1$, ensuring $\Pr[e_v]=\Pr[e_{v'}]$, $\forall v,v' \in \V$.

\begin{theorem}
    Let $\qdist$ be a query distribution over $\queries$.
    Let $\qdist'$ be the output of Algorithm \ref{algo:singleton} with input $\qdist$, then we have:
\begin{displaymath}
\Pr_{\qdist'}[v] = \Pr_{\qdist'}[v'] \text{ for all } v \in \V.
\end{displaymath}
\label{theorem:flat_singleton}
\end{theorem}

Algorithm~\ref{algo:singleton} is a partial step toward $\qdist^*$, but our findings show that flattening frequencies beyond singletons is impossible, ruling out any hope of identifying such a $\qdist^*$.
In fact, Theorem~\ref{unequal_distance} shows that no expressive query distribution can assign equal retrieval probability to value pairs of differing $L^1$ distance.

\begin{theorem}
\label{unequal_distance}
    Let $\qdist$ be any query distribution over the universe of queries $\queries$ and domain $\V = [N]^k, N>2$, such that every query $\query \in \queries$ has non-zero probability.
    For every pair of values $v,v'$ in $\V$ with $L^1$-distance $dist(v,v')=d$, there exists a pair of values $v,v''$ with $L^1$-distance $dist(v,v'') \neq d$ such that $\Pr[v \cap v'] \neq \Pr[v \cap v'']$.
\end{theorem}

\begin{figure}[t]
    \centering
    \includegraphics[scale=0.41]{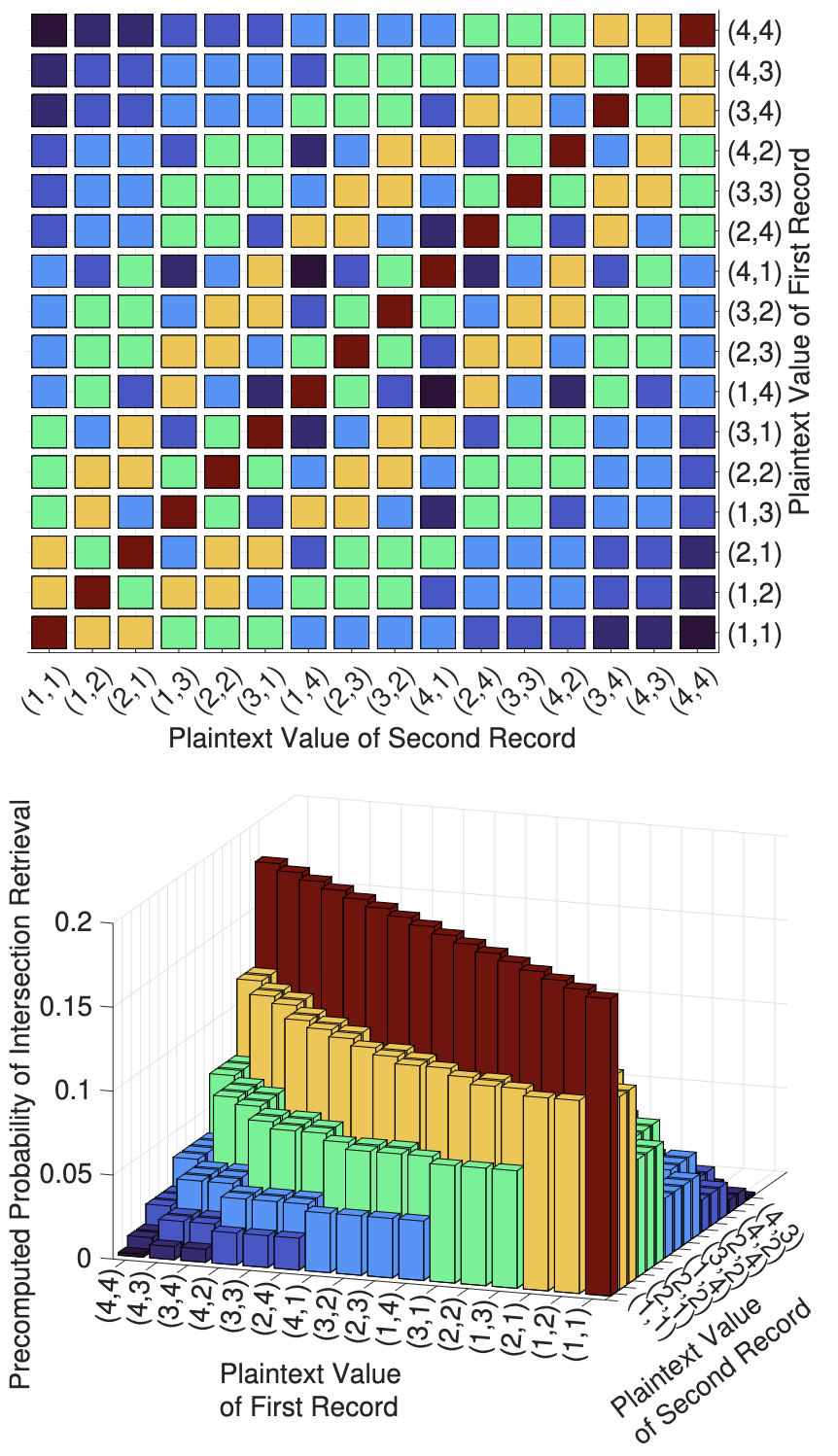}
    \caption{An illustration  of how the query distribution by Algorithm \ref{algo:equidistant} affects the frequency of retrieval of pairs (in $\values$=$[4]\times[4]$). All pairs $(v,v')$ that have a fixed $L^1$-distance, have the same probability of retrieval $\Pr_{\widehat{\qdist}}[v \cap v']$.}
    \label{fig:flat_pairs}
\end{figure}

\subsection{All Equidistant Value Pairs Can Have the Same Probability of Retrieval}
 \label{sec:flatten_tuples}

Fortunately, Theorem~\ref{unequal_distance} leaves room for a weaker but meaningful alternative, namely a reconstruction space (implied by a query distribution) under which all value pairs at the same $L^1$ distance share the same retrieval probability.
Algorithm~\ref{algo:equidistant} (in Appendix~\ref{sec:appendix_algo}
) constructs exactly this by iterating Algorithm~\ref{algo:singleton} across all distances $d=0,1,\ldots,k(N-1)$, where distance $0$ corresponds to a value paired with itself.

\begin{theorem}
    Let $\qdist$ be an expressive query distribution over $\queries$.
    The distribution $\widehat{\qdist}$ output by Algorithm~\ref{algo:equidistant} on input any $\qdist$ equalizes retrieval probabilities across equidistant pairs, that is, for all pairs $(v,v'),(v'',v''')$ with $dist(v,v')=dist(v'',v''')$:
    \begin{displaymath}
        \Pr_{\widehat{\qdist}}[v \cap v']=\Pr_{\widehat{\qdist}}[v'' \cap v'''].
    \end{displaymath}
    \label{theorem:alg2}
\end{theorem}

Figure~\ref{fig:flat_pairs} illustrates the effect of $\widehat{\qdist}$ on a toy example over $\values=[4]\times[4]$, where axes represent the plaintext values of each record in the pair.
All pairs at the same $L^1$ distance share the same retrieval probability (indicated by a unique color), so an attacker observing co-retrieval frequencies learns the pairwise distances but gains no advantage in recovering the true values.
More generally, $\rs{\widehat{\qdist}}(\db)$ consists of all databases $\db'$ over $\rids,\V$ in which every pair of records preserves its $L^1$ distance from $\db$; the attacker learns pairwise distances, and nothing more.

\textbf{Comparison with \texttt{PANCAKE~\cite{DBLP:conf/uss/GrubbsKLBL0R20}.}} The analysis in this section presents some resemblance to PANCAKE (which also manipulates access frequencies to resist frequency-matching attacks), but differs in two fundamental ways: 
\emph{(1) Setting.} PANCAKE targets key-value stores where every query retrieves exactly one record, whereas our setting involves range queries that may retrieve anywhere from one to all $n$ records simultaneously. This makes flattening significantly harder in our setting, as co-retrieval frequencies across tuples of records must be equalized, not just individual retrieval frequencies. 
\emph{(2) Mechanism:} PANCAKE achieves flattened single-record retrievals by injecting fake accesses and replicating records at the system level; our analysis aims to identify a native \emph{query distribution} directly, without any fake queries or data replication.

 \section{Evaluation}
\label{sec:evaluation}
\sethlcolor{yellow}

\begin{figure}[hbp]
    \centering
    
    \includegraphics[width=\linewidth]{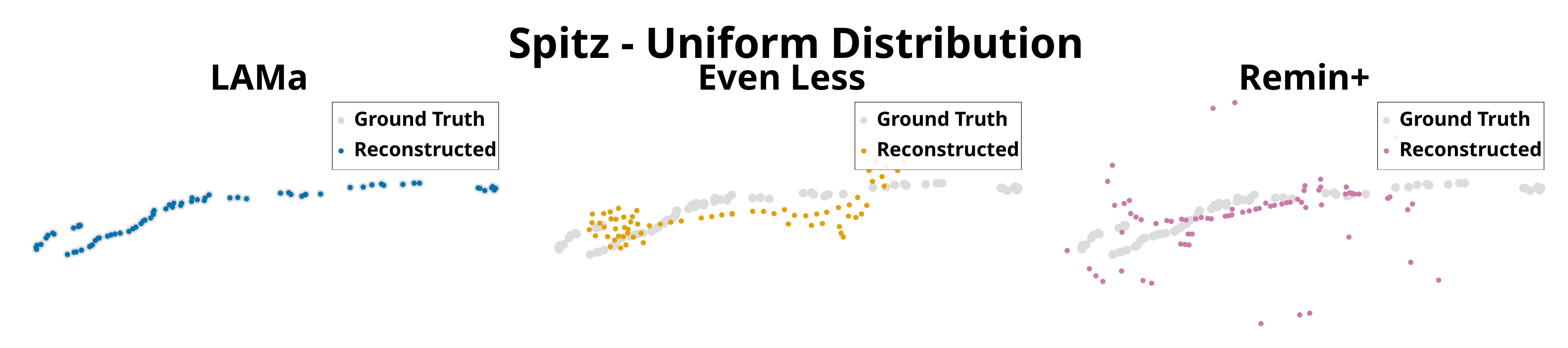}
    
    \includegraphics[width=\linewidth]{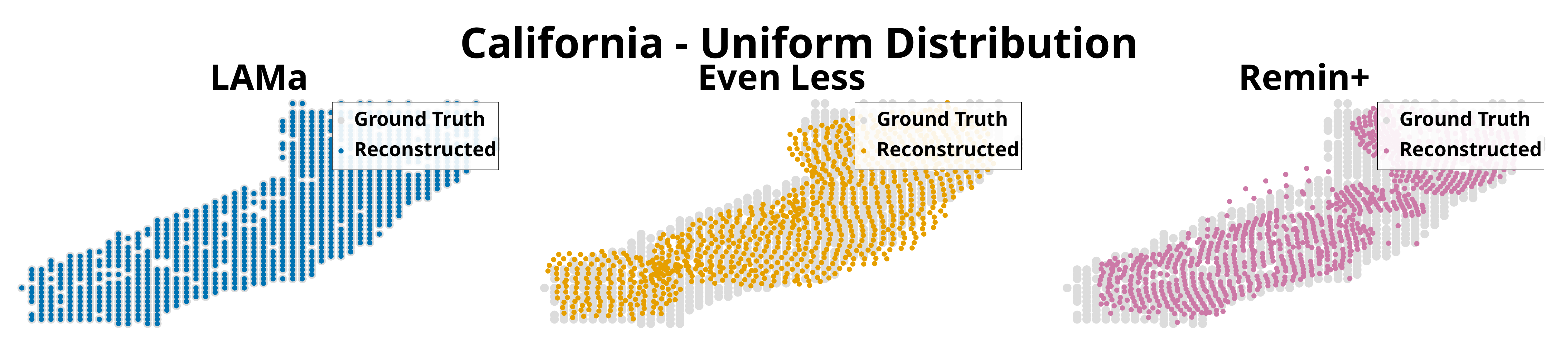}
    
    \includegraphics[width=\linewidth]{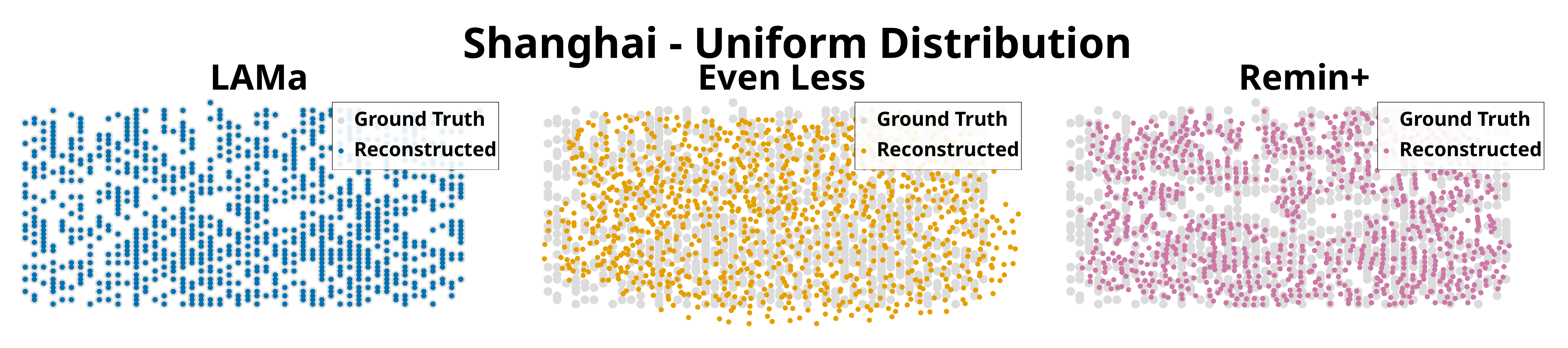}
    
    \includegraphics[width=\linewidth]{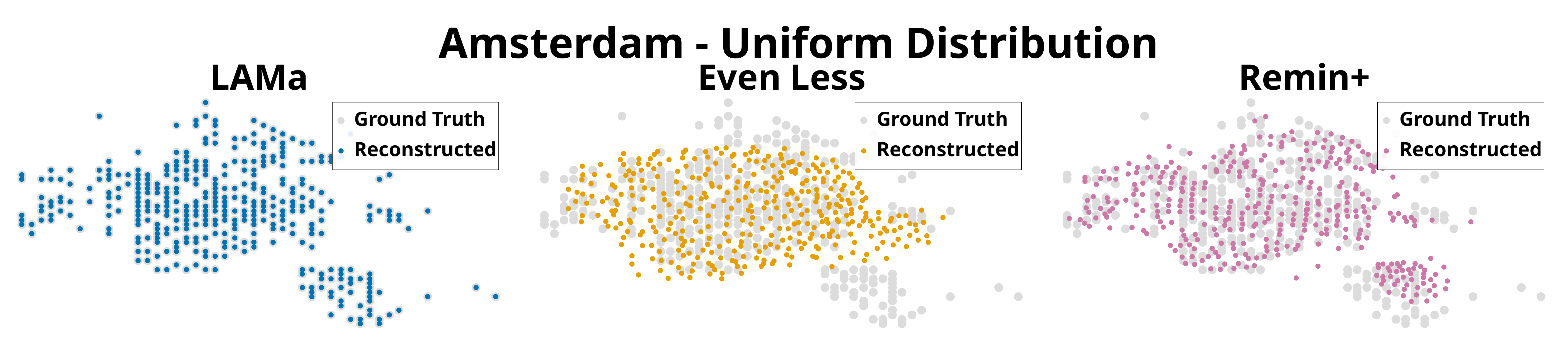}
    
 \includegraphics[width=\linewidth]{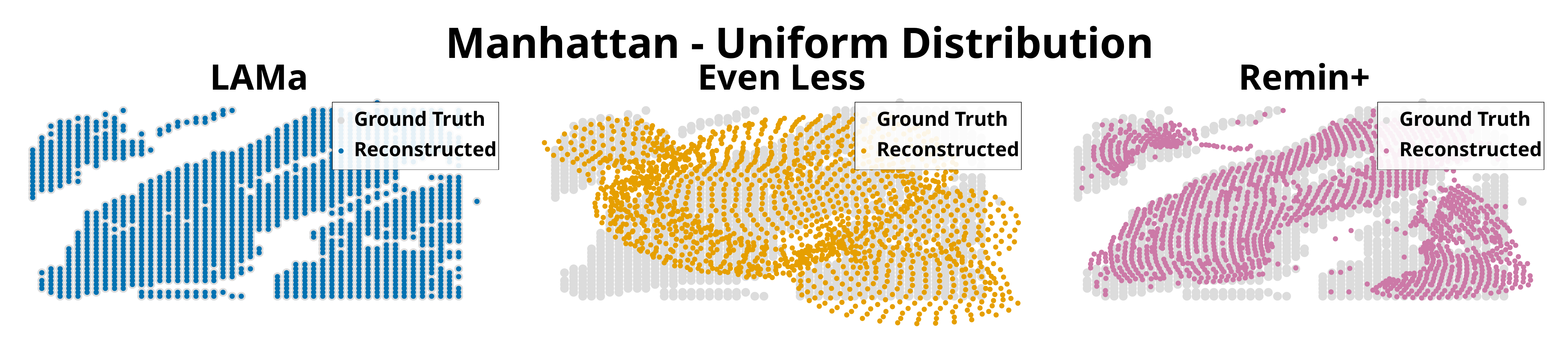}
    
    \includegraphics[width=\linewidth]{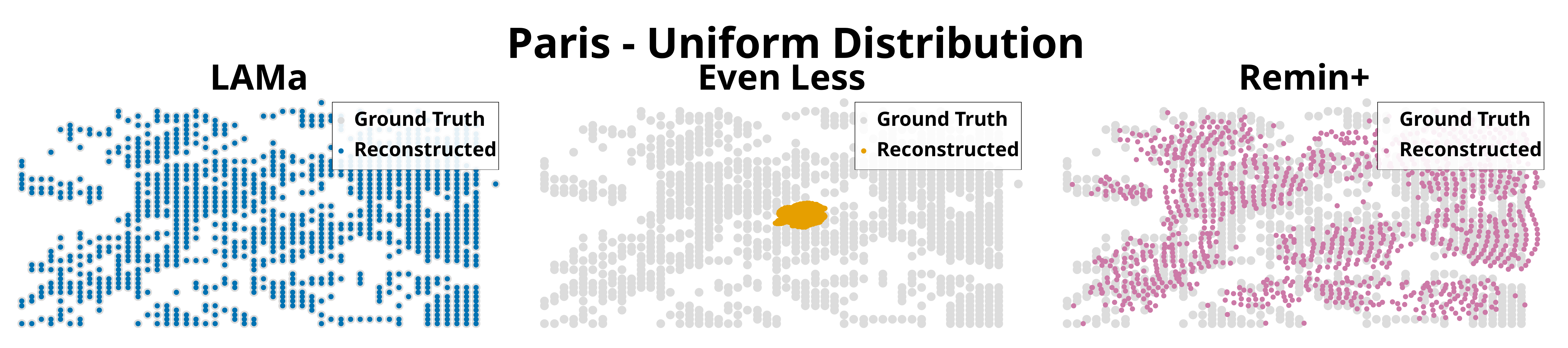}
    
    \caption{Comparison of reconstructions across multiple datasets. (Spitz, California, Shanghai, Amsterdam, Manhattan, and Paris). For each dataset, we show the best possible reconstruction.}
    \label{fig:reconstruction_comparison}
\end{figure}

\begin{figure}[htbp]
    \centering
    
    \includegraphics[width=\linewidth]{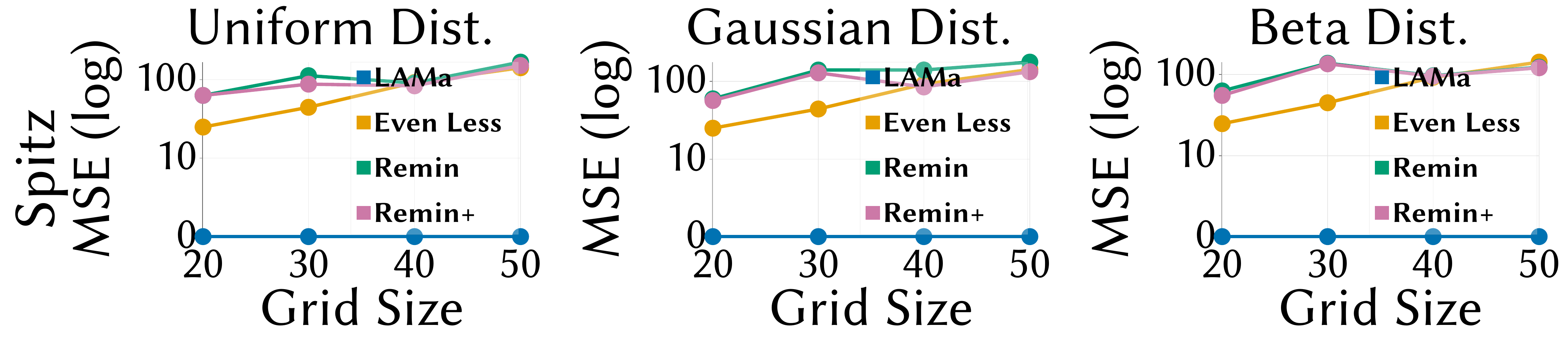}
    
    \includegraphics[width=\linewidth]{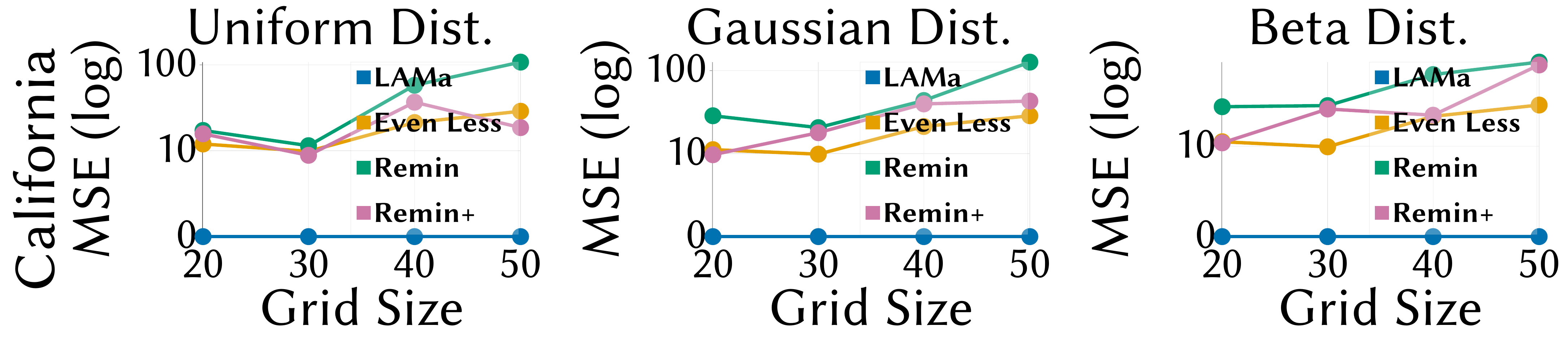}
    
    \includegraphics[width=\linewidth]{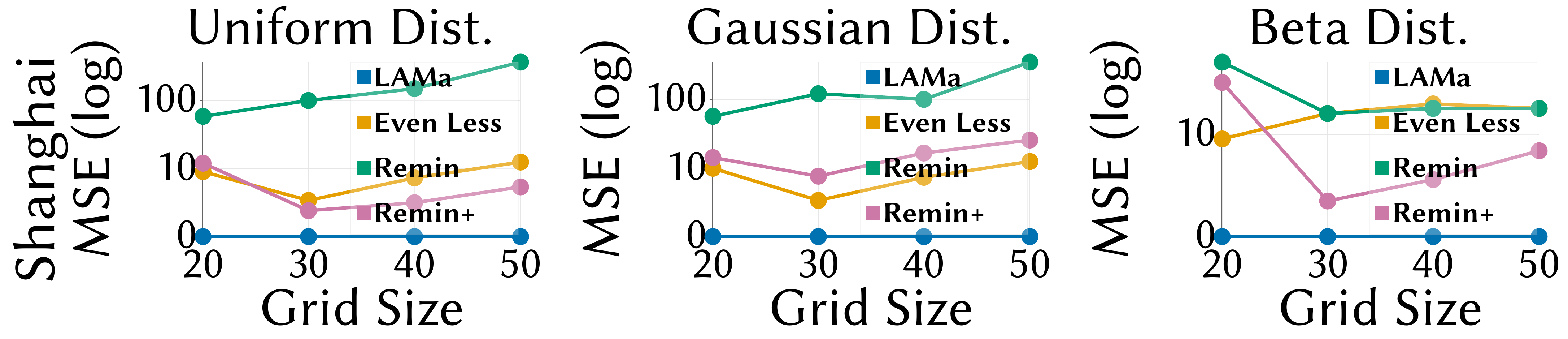}
    
    \includegraphics[width=\linewidth]{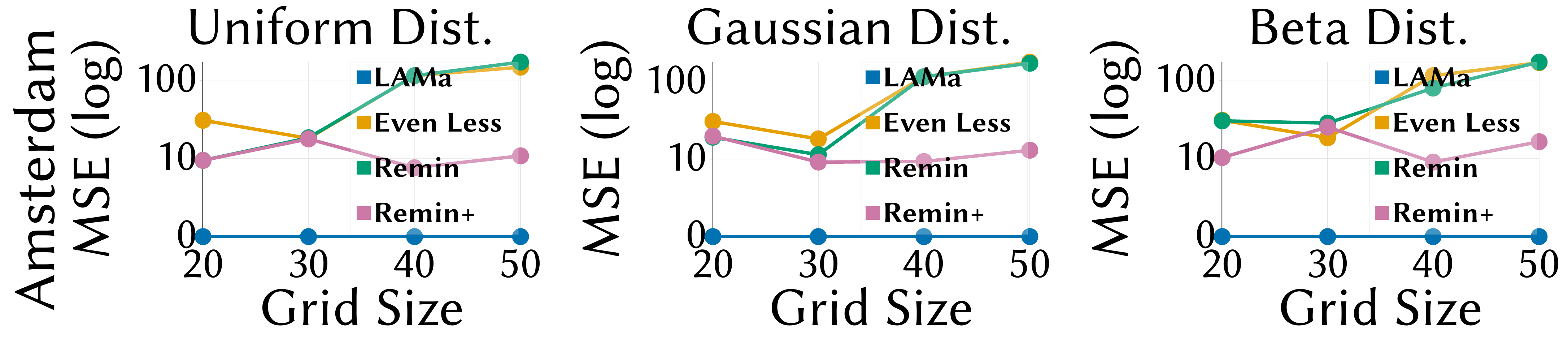}

    \includegraphics[width=\linewidth]{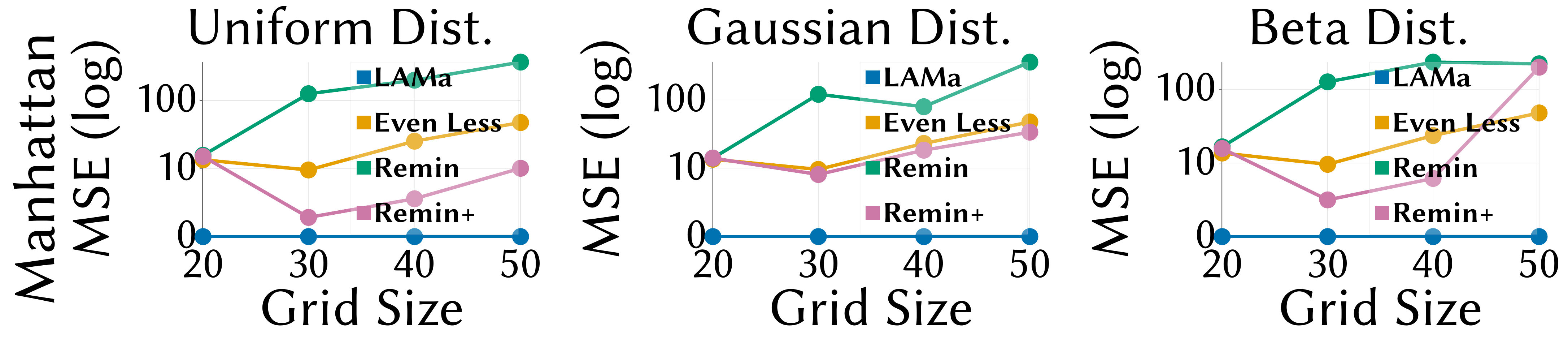}
    
    \includegraphics[width=\linewidth]{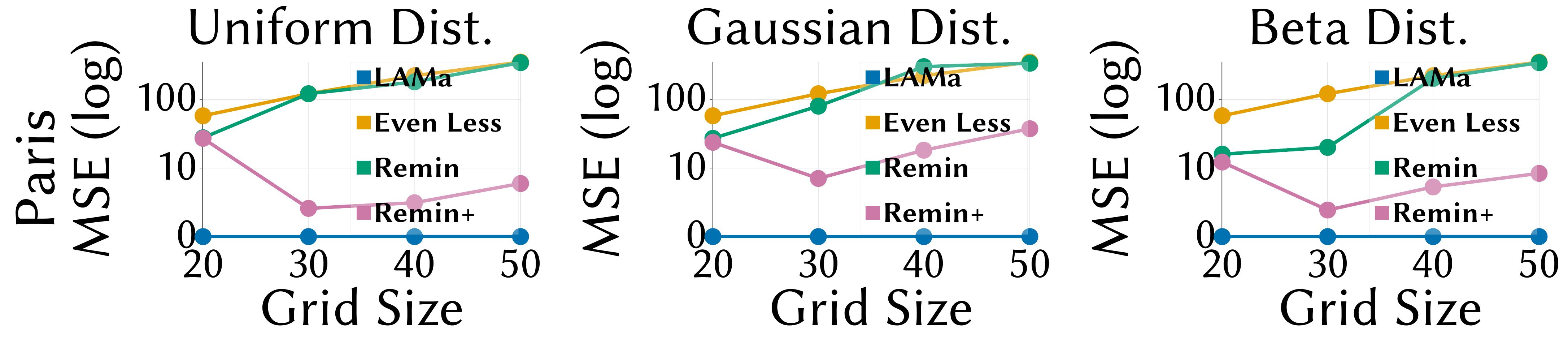}

    \caption{Mean Squared Error of best found reconstruction measured at different grid scales and query distributions.}
    \label{fig:mse_grid_size}
\end{figure}

\begin{figure}[htbp]
    \centering
    
    \includegraphics[width=0.9\linewidth]{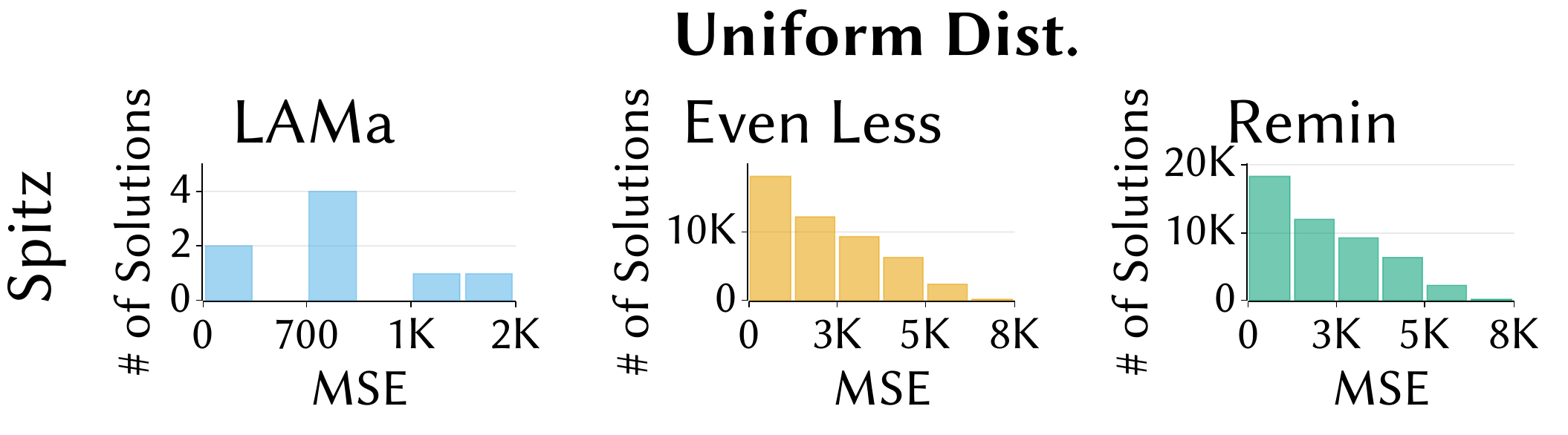}
    
    \includegraphics[width=0.9\linewidth]{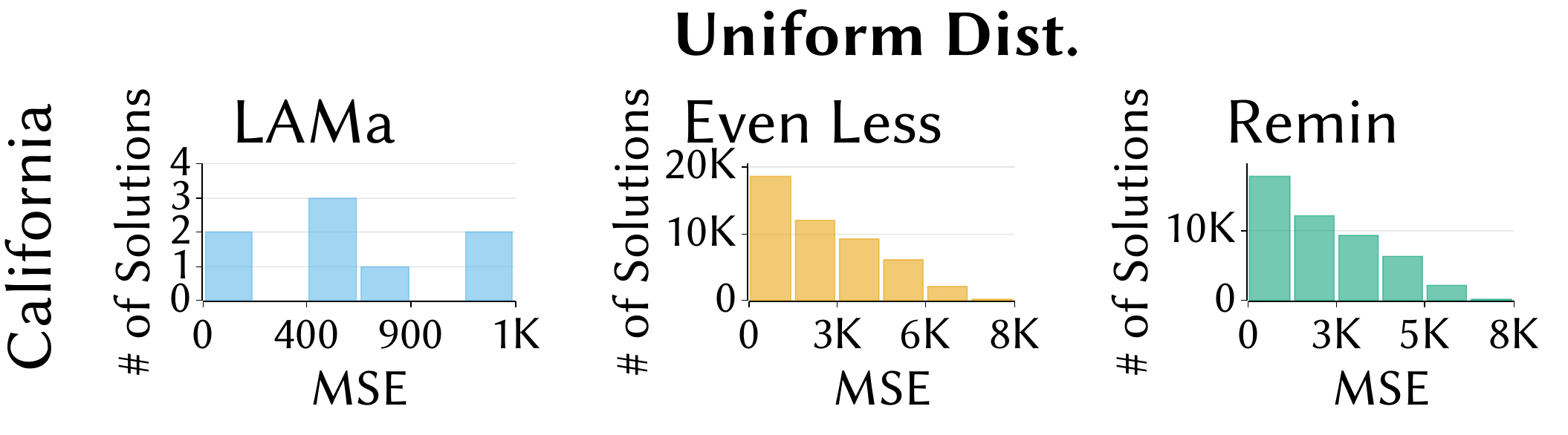}
    
    \includegraphics[width=0.9\linewidth]{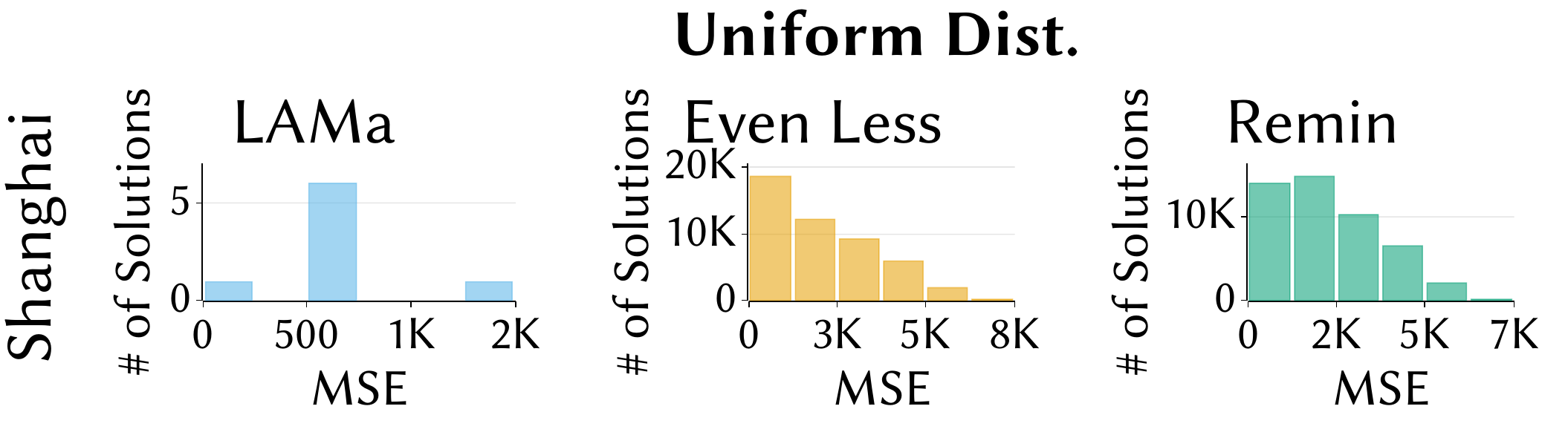}
    
    \includegraphics[width=0.9\linewidth]{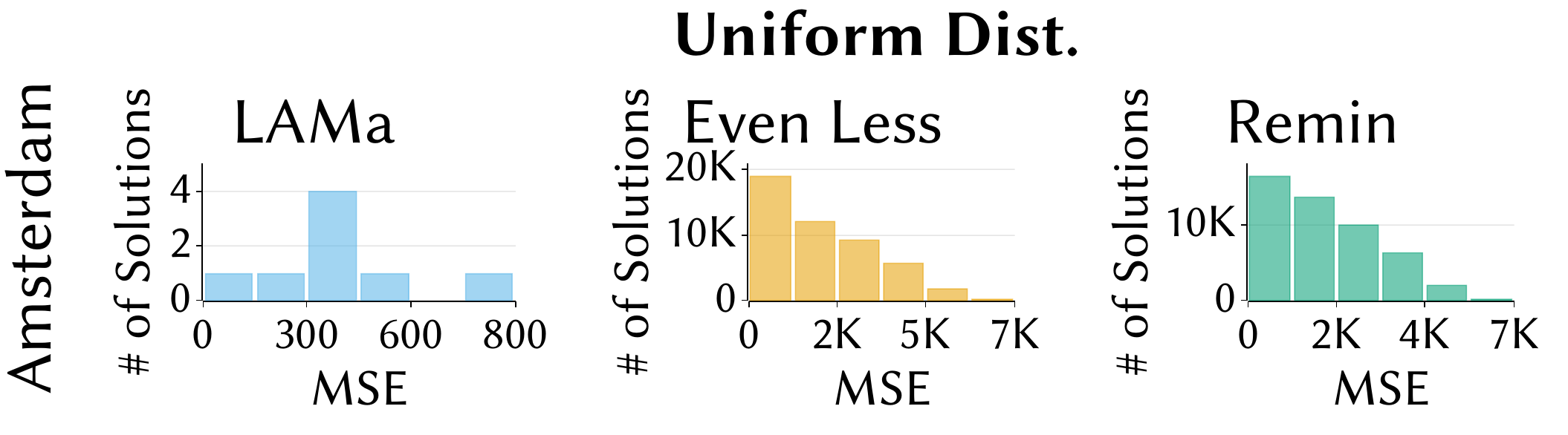}

\includegraphics[width=0.9\linewidth]{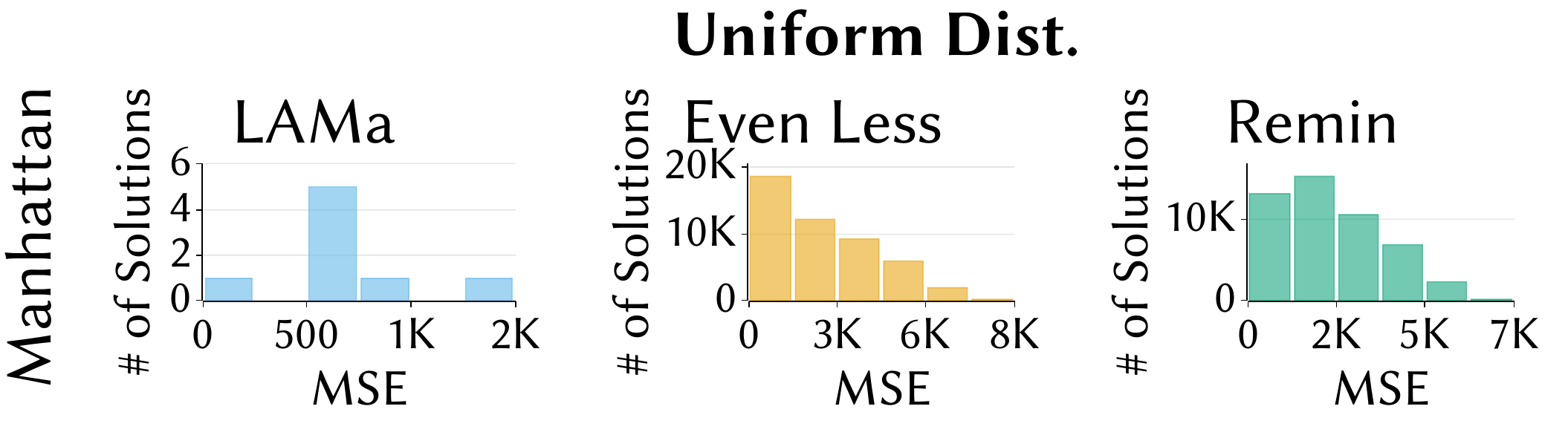}
    
    \includegraphics[width=0.9\linewidth]{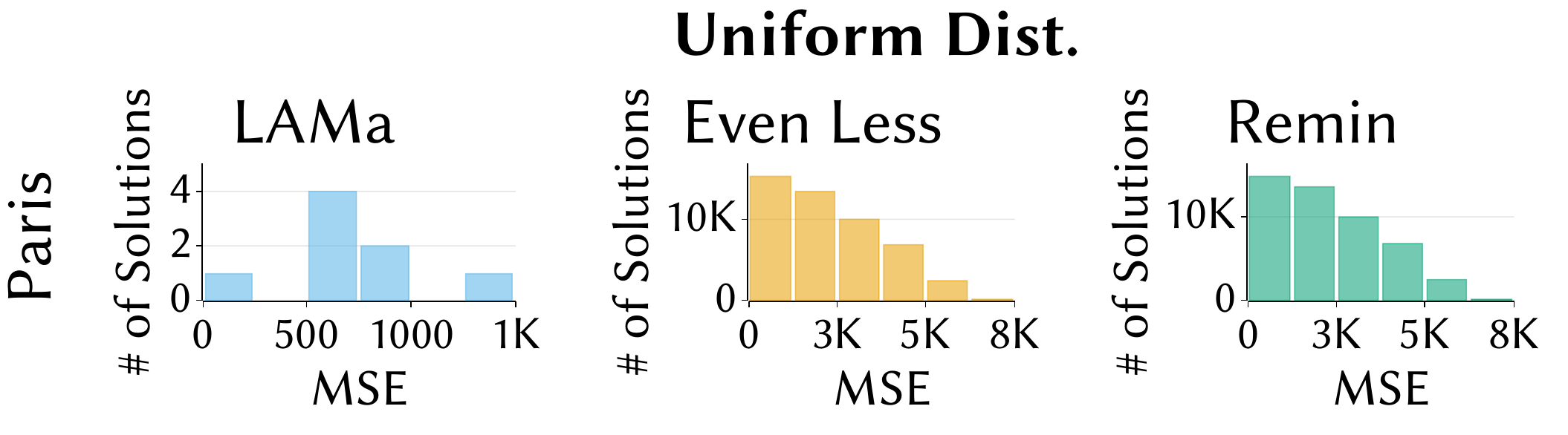}
    
    \caption{Histograms of all possible solutions for each method and their respective Mean Squared Error. Plots use a uniform query distribution.}
    \label{fig:mse_histogram}
\end{figure}

We implement \framework to: \textbf{(1)} Evaluate \framework in the ideal case for the attacker in which the frequencies are exact; for this setting, the state-of-the-art attacks \texttt{Even Less}, \texttt{Remin}, and \texttt{Remin+}~\cite{even_less,remin}, referred to as SOTA hereafter, are also allowed their ideal case, meaning, observing all possible responses, 
and injecting arbitrary plaintexts, \textbf{(2)} Evaluate the effectiveness of \framework in the ``real world'', where observed frequencies do not match exactly the probabilities of retrieval, i.e., only a percentage of queries are observed,
\textbf{(3)} Evaluate \framework's performance on the query distribution introduced in Section~\ref{sec:flattening} that increases the reconstruction space. 
The datasets, query distributions,  and grid-sizes used in our evaluation are taken \textbf{directly from prior work}~\cite{even_less,remin} on multidimensional attacks.
Omitted plots can be found in Appendix~\ref{sec:appendix_plots}.

\textbf{Testbed.} 
All experiments were compiled on a Debian 12 container with a 64-core 128-thread AMD EPYC ZEN 2 CPU. 
For every experiment on 2-dimensional data, we use 512GB of RAM, and for 3 dimensions and above, we use 3TB of RAM. 
\framework is implemented and compiled using Rustc v1.95.0 and uses Google Optimization Tools (OR-Tools) V9.15 to solve constraint satisfaction problems. 
All results for \texttt{Remin}, \texttt{Remin+} \cite{remin} and \texttt{Even Less} \cite{even_less} were generated using their respective Python codebases from the original authors.

\textbf{Datasets.} We use the same datasets as~\cite{remin} alongside the Spitz dataset~\cite{spitz}.
\emph{Spitz}: 6 months of mobile phone records of politician Malte Spitz~\cite{spitz}; only November 1--15 is used.
\emph{Grid}: A $d$-dimensional grid with user-specified sparsity and size, used as a stand-in for higher-dimensional datasets~\cite{remin}.
\emph{California}: 21,000 road intersections in California, from~\cite{cali_map}.
\emph{Amsterdam}: Water access locations in Amsterdam, collected via OpenStreetMaps by~\cite{remin}.
\emph{Manhattan}: Highway crossings across Manhattan, created by~\cite{remin}.
\emph{Paris}: Shop locations across central Paris, created by~\cite{remin}.
\emph{Shanghai}: Bus stops in downtown Shanghai, created by~\cite{remin}.
We follow the methodology of previous work~\cite{remin,even_less} and scale 2D data to a 50x50 grid unless specified otherwise.
\emph{New Hampshire}: A $16\times16\times14$ 3D terrain dataset of the White Mountains, taken from~\cite{new_hampshire}.

\subsection{Ideal-Case Evaluation for the Attacker}

In the first part of our evaluation, we consider the ideal-case scenario for the attacker; that is, \framework uses exact frequencies, i.e., equalities hold in Equation~(\ref{eq:equality}), while the SOTA attacks~\cite{even_less, remin} have observed all possible responses, and~\cite{remin} additionally allows plaintext injection into the database under attack.

\textbf{Visual Comparison.} Figure~\ref{fig:reconstruction_comparison} compares the spatial reconstruction quality of \framework, \texttt{Even Less}~\cite{even_less}, and \texttt{Remin+}~\cite{remin} across five geographic datasets under a uniform query distribution, showing the best reconstruction produced by each technique. \framework produces \emph{perfect reconstructions}, fully overlapping with the ground truth across all datasets. 
Both SOTAs achieve reasonable reconstructions, but can degrade on sparse and complex datasets. 
The sparsity of Paris causes \texttt{Even Less} to collapse its reconstructions into a small concentrated region, a consequence of their graph-based approach connecting geometrically distant points with edges and the subsequent graph drawing step pulling them falsely close together.    
\emph{Only \framework maintains consistent reconstruction as spatial sparsity and geometric complexity increase.}

\textbf{Mean Squared Error-Based Comparison.} Figure~\ref{fig:mse_grid_size} reports the MSE (logarithmic scale) of each technique across five datasets, three query distributions (Uniform, Gaussian, Beta), and grid sizes from 20 to 50, showing the best reconstruction per technique. \framework achieves zero MSE across all datasets and configurations, with its curve remaining flat. Both \texttt{Even Less} and \texttt{Remin} incur \emph{MSE up to two orders of magnitude higher} than \framework, growing monotonically with grid size, demonstrating the fundamental limitations of graph-based reconstruction.
\texttt{Remin+} outperforms the other two SOTAs yet still falls short of \framework by orders of magnitude on more complex datasets. 
The gap between \framework and the SOTAs is consistent across all three query distributions, confirming that \framework's advantage is not specific to any distribution or dataset.

\textbf{Comparing the Reconstruction Space.} Figure~\ref{fig:mse_histogram} shows the histogram of MSE (based on a single run) across all possible solutions per technique under a uniform query distribution. 
For the SOTAs~\cite{even_less,remin}, we enumerate all combinations of rotation, scaling, and translation applied to each graph embedding; note that the attacker has \emph{no way of identifying which parameterization yields the best reconstruction}\footnote{As opposed to the experiment of the previous paragraph in which we hand-picked the best reconstruction; an evaluation methodology used in prior work~\cite{even_less,remin}.}, and must therefore pick from a vast reconstruction space. \framework produces 8 tightly concentrated set of solutions with low MSE, while \texttt{Even Less} and \texttt{Remin} each generate tens of thousands of solutions with MSE many times larger than \framework's worst case.
We exclude  \texttt{Remin+} from this comparison as \texttt{Remin+} uses plaintext injection (a much more advantageous but unrealistic threat model) to find a single reconstruction with optimal alignment. 
Nevertheless, \texttt{Remin+} incurs non-zero error: $5.45$ (Shanghai), $5.96$ (Paris), $18.72$ (California), $11.01$ (Amsterdam), $10.24$ (Manhattan), and $151.02$ (Spitz).

\subsection{Realistic-Case Evaluation for the Attacker}
\label{sec:realistic_evaluation}

In the second part of our evaluation, we consider the real-case scenario; that is, the attacker observes $10\%$, $20\%$, $30\%$ of available queries (or ``query percentages'' from~\cite{remin}) to approximate the frequencies, with ($\epsilon,\delta$)-guarantees, much like Equation~(\ref{eq:eps_bound}). 
The SOTA attacks~\cite{even_less, remin} observe the same number of queries, and~\cite{remin} additionally allows plaintext injection into the database under attack. 
Due to the volume of experiments, we use a $15\times15$ grid in this section; all trends have been verified to hold on larger grid sizes.

\begin{figure}[htbp]
    \centering
    
     \includegraphics[width=\linewidth]{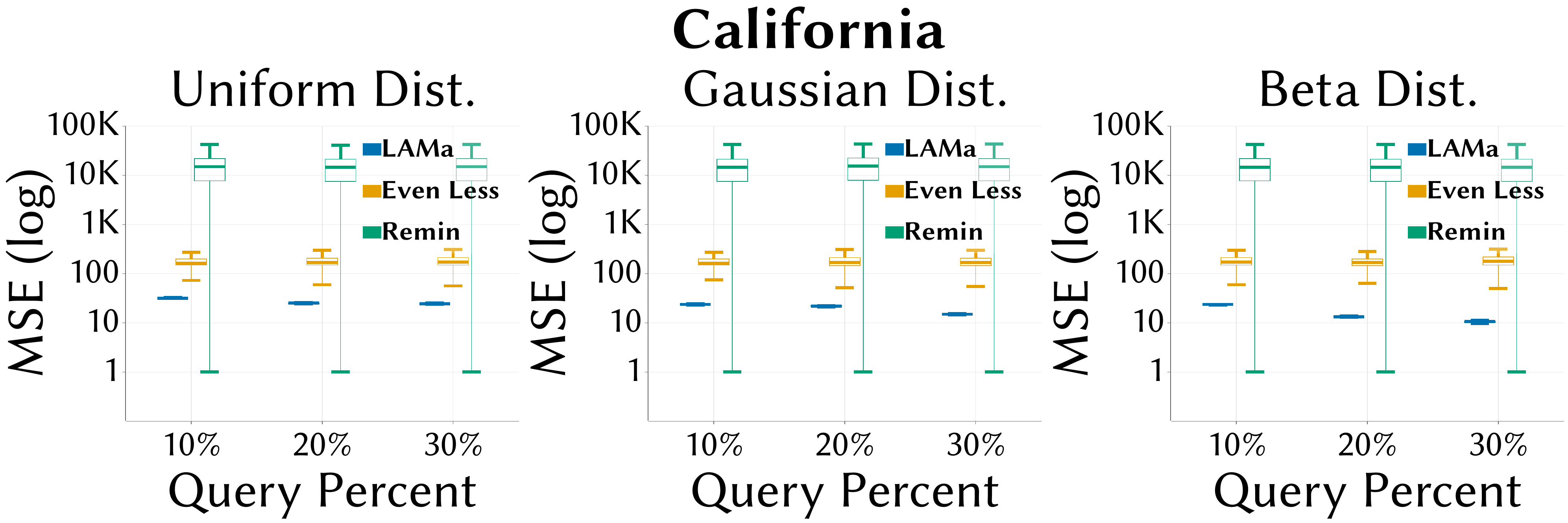}
    
    \includegraphics[width=\linewidth]{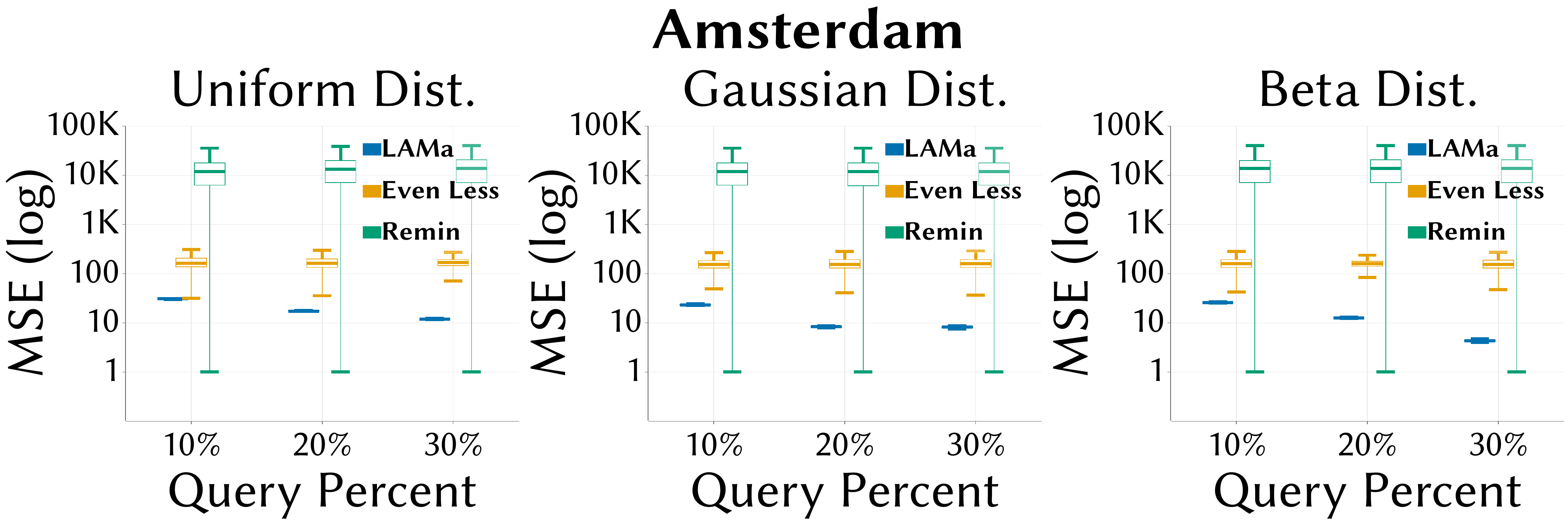}
    
     \includegraphics[width=\linewidth]{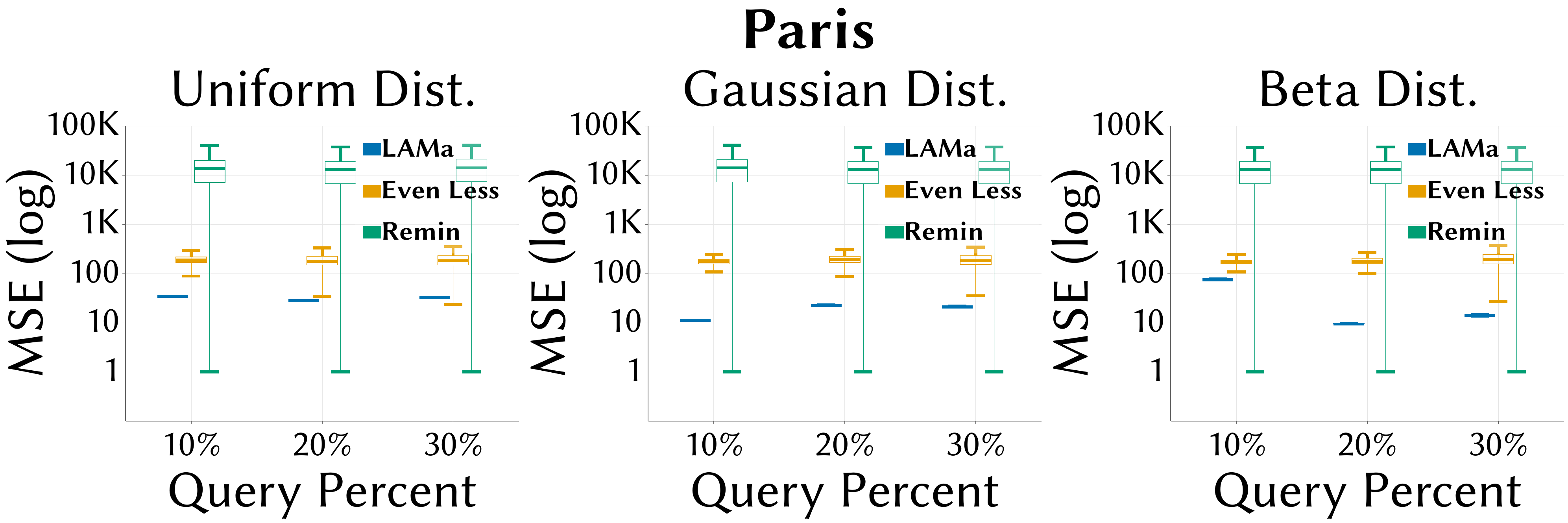}
    
    \includegraphics[width=\linewidth]{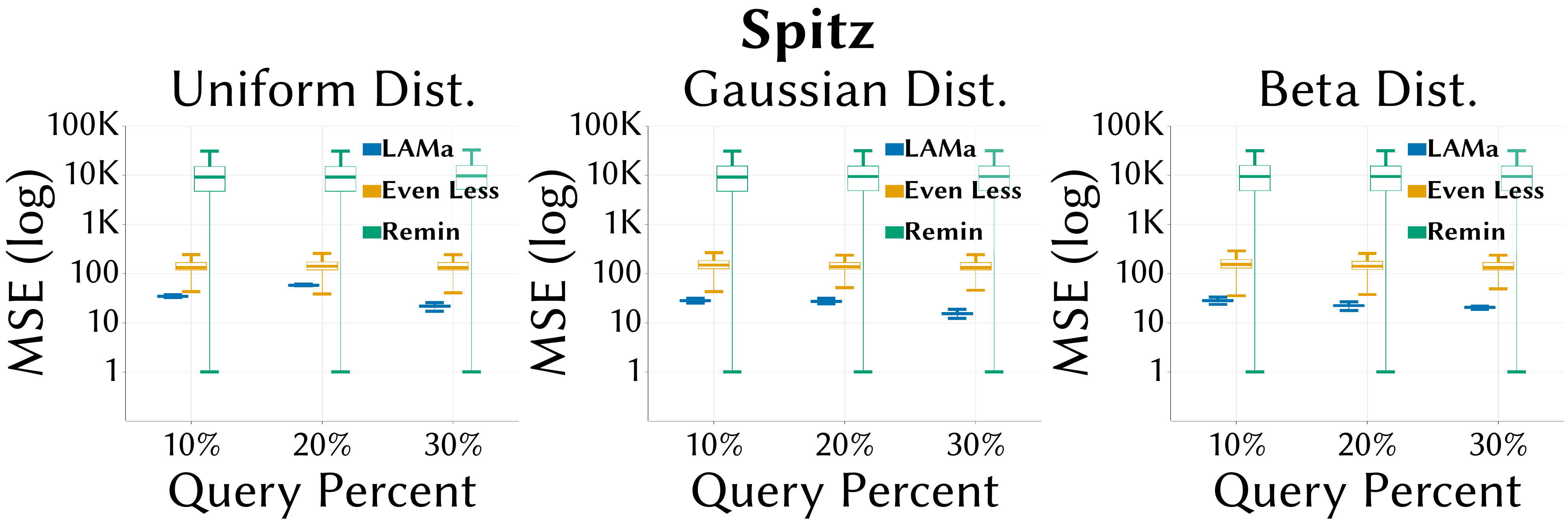}

        \caption{MSE across sampled reconstructions at different query percentages.}
    \label{fig:big_combined_figure}
\end{figure}

\textbf{MSE Across Sampled Reconstruction.} Figure~\ref{fig:big_combined_figure} reports MSE distributions via boxplots across California, Amsterdam, Paris, and Spitz under three query distributions and query percentages of 10\%, 20\%, and 30\%. 
For each technique, we generated $10$K reconstructions and measured their MSE from the ground truth. 
At higher query percentages, \framework's median MSE is up to \emph{an order of magnitude lower} than both SOTAs~\cite{even_less, remin} across all datasets and distributions, with a tightly concentrated box reflecting low variance. Crucially, \framework's median MSE decreases (in most cases) as the query percentage increases, confirming that more observed queries give a more accurate view of the reconstruction space. 
In contrast, \texttt{Even Less} and \texttt{Remin} maintain large, high-variance boxes across all query percentages, with no consistent improvement as more queries are observed, reflecting the fundamental limitations of graph-based reconstruction under realistic observation conditions.

    \begin{figure}[htbp]
    \centering

     \includegraphics[width=\linewidth]{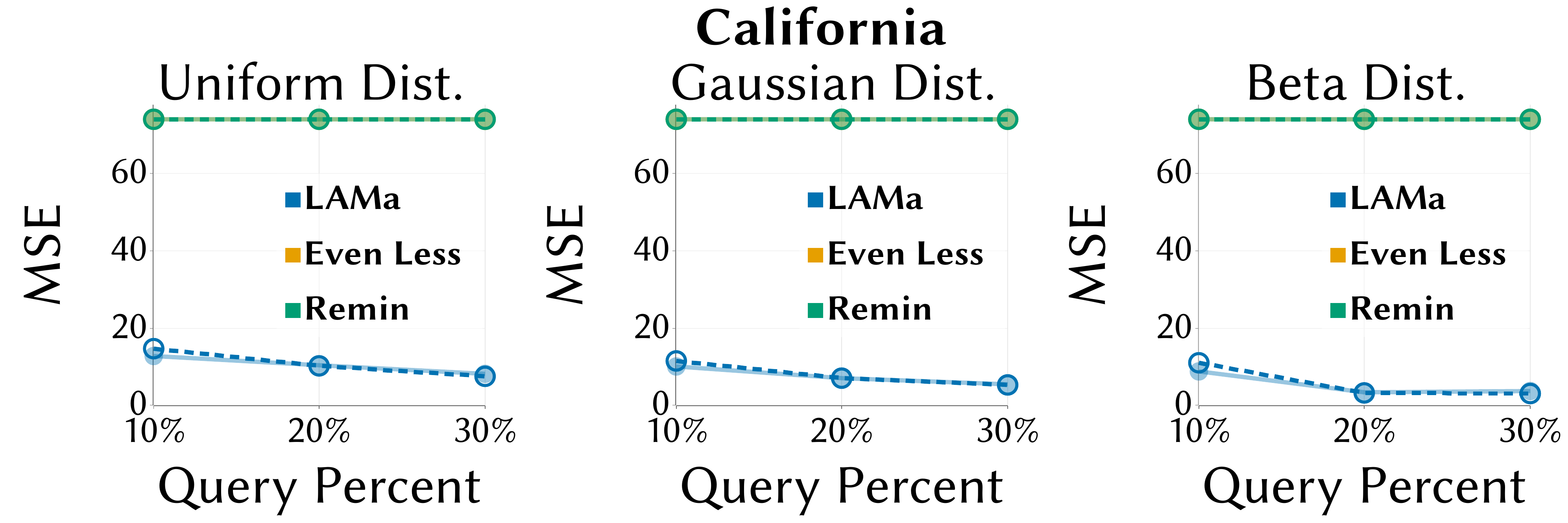}
    
    \includegraphics[width=\linewidth]{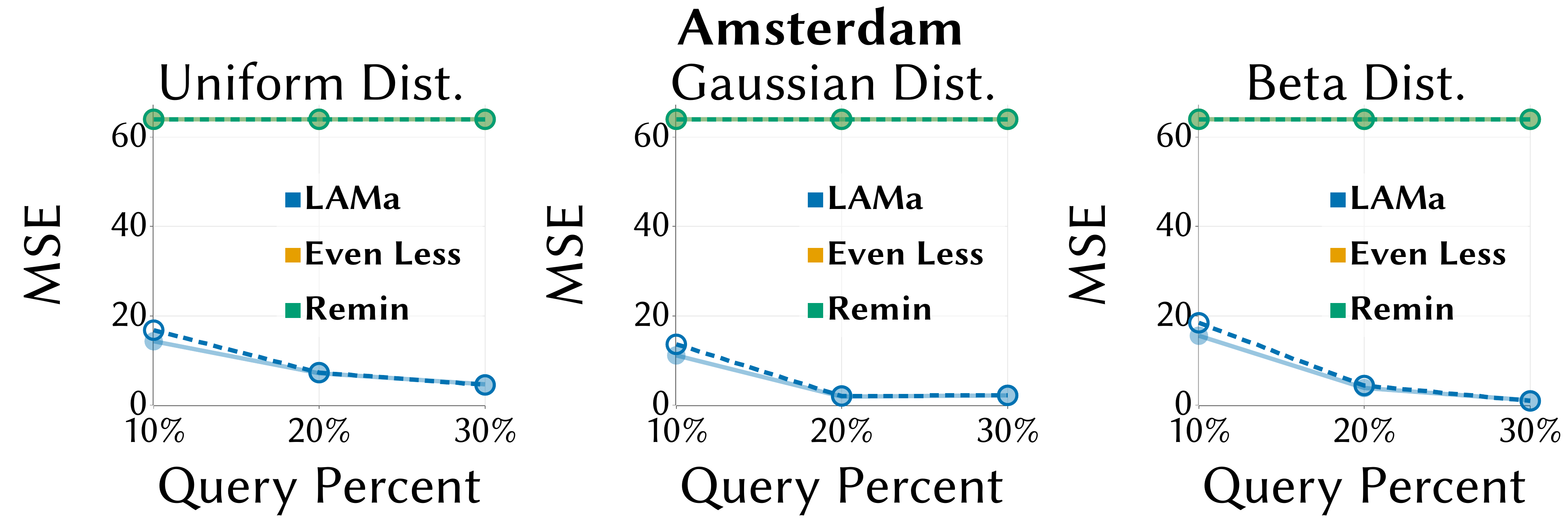}
    
     \includegraphics[width=\linewidth]{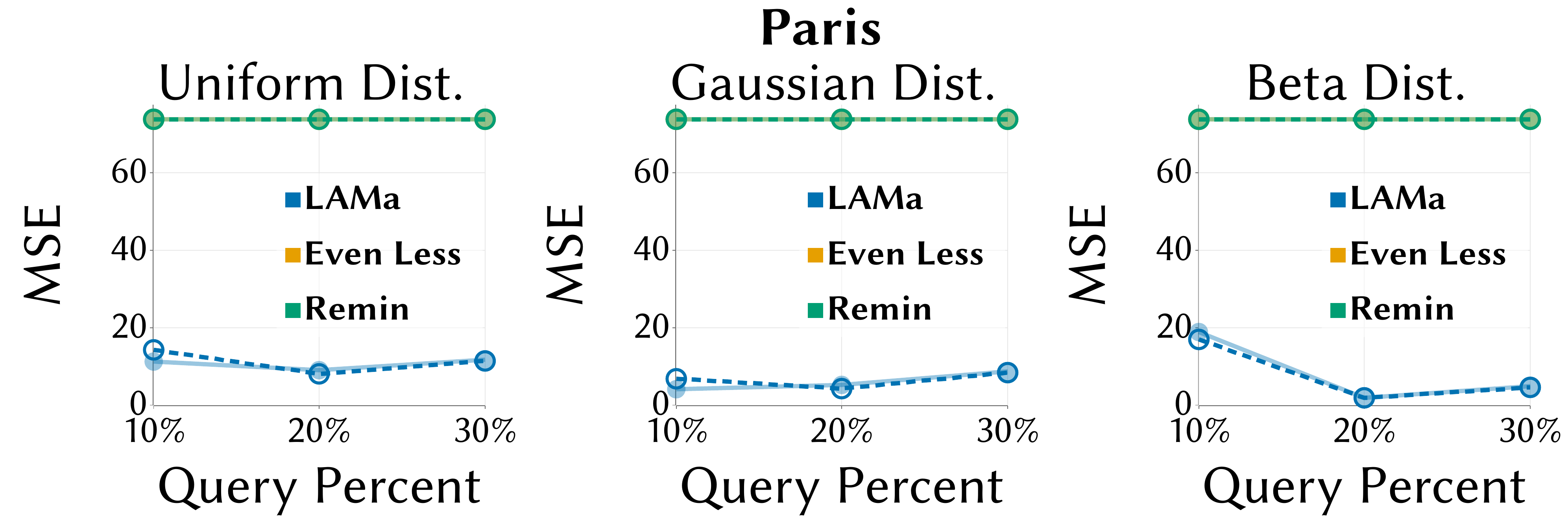}
    
    \includegraphics[width=\linewidth]{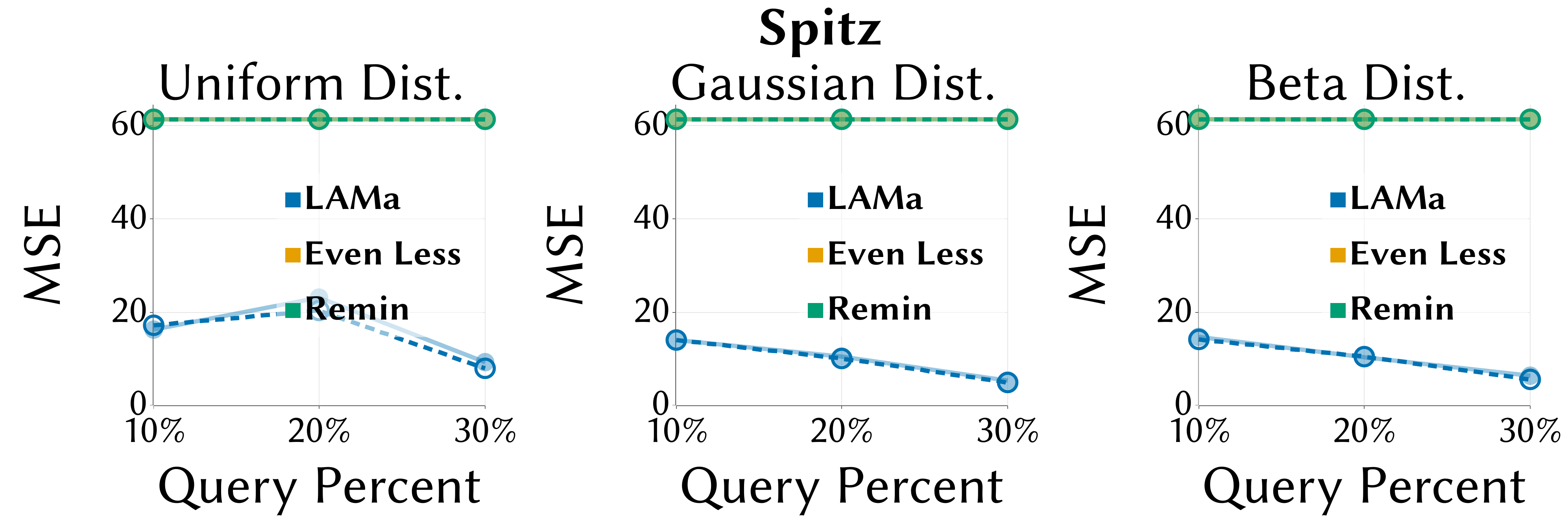}
    
\caption{MSE of output quality guarantees derived from the sampled reconstruction space. Dashed line: plaintext point placed at the center of the largest diameter across all its reconstructions. Solid line: plaintext point placed at the centroid across all its reconstructions.}
    \label{fig:worst_case}
\end{figure}

\textbf{Output Quality Guarantees.} Figure~\ref{fig:worst_case} reports the MSE of output quality guarantees (averaged across ten runs) derived from the sampled reconstruction space, across California, Amsterdam, Paris, and Spitz under three query distributions and query percentages of 10\%, 20\%, and 30\%. 
We emphasize that we apply the output quality guarantee framework introduced in Section~\ref{sec:output} to both \framework and the SOTAs~\cite{even_less,remin}, noting that \emph{these SOTAs provide no such formal guarantees in their original work}.  
For each method, we report the two output strategies from Section~\ref{sec:output}---the \texttt{diameter} approach (dashed line) and the \texttt{centroid} approach (solid line)---each committing to a single reconstructed database after deriving several initial solutions/reconstructions. 
Both SOTAs~\cite{even_less, remin} remain flat across all query percentages and datasets (due to symmetries of moving graph embeddings), showing no improvement as more queries are observed. In contrast, both placement strategies for \framework yield MSE that decreases consistently as the query percentage increases, remaining well below the SOTAs across all configurations, confirming that our theoretical output quality guarantees translate directly into empirical improvements. 
We also report the maximum $\epsilon$---the worst-case distance between a single reconstructed point and the ground truth---in Table~\ref{tab:output_quality}. \framework outperforms SOTA across all cases.

\begin{table}[h]
\centering
\setlength{\aboverulesep}{0pt}\setlength{\belowrulesep}{0pt}
\resizebox{\columnwidth}{!}{%
\begin{tabular}{l|c|c|c|c|c|c|c|c|c|c|c|c|}
\cline{2-13}
& \multicolumn{2}{c|}{\rotatebox{60}{Spitz}} 
& \multicolumn{2}{c|}{\rotatebox{60}{California}} 
& \multicolumn{2}{c|}{\rotatebox{60}{Shanghai}} 
& \multicolumn{2}{c|}{\rotatebox{60}{Amsterdam}} 
& \multicolumn{2}{c|}{\rotatebox{60}{Manhattan}} 
& \multicolumn{2}{c|}{\rotatebox{60}{Paris}} \\
\cline{2-13}
& cent & diam & cent & diam & cent & diam & cent & diam & cent & diam & cent & diam \\
\hline
\hline
\framework & 5.47 & 4.72 & 4.24 & 3.51 & 2.37 & 1.81 & 3.43 & 3.01 & 3.35 & 2.72 & 4.20 & 3.25 \\
\hline
\texttt{Even Less} & 20.14 & 20.14 & 21.88 & 21.88 & 20.15 & 20.15 & 21.75 & 21.75 & 21.93 & 21.93 & 23.61 & 23.61 \\
\hline
\texttt{Remin} & 3877.69 & 3877.69 & 65.04 & 65.04 & 140.01 & 140.01 & 61.55 & 61.55 & 132.12 & 132.12 & 134.22 & 134.22 \\
\bottomrule
\end{tabular}
}
\caption{Output quality $\epsilon$ for \texttt{centroid} and \texttt{diameter} placement at 30\% query percentage. 
}
\label{tab:output_quality}
\end{table}

\textbf{Higher-Dimensional Databases.} In Figure~\ref{fig:high-d} we demonstrate the reconstruction of \framework with 30$\%$ of queries under Gaussian query distribution. Similar to other SOTAs we evaluate \framework's performance on higher dimensions by using the ``Grid'' dataset~\cite{remin}. A direct comparison was not available because of SOTAs codebase limitations. 
In our experiments, we measured the MSE of \framework across dimensions 2 through 6 under a Gaussian query distribution, using the centroid reconstruction method. 
\framework achieves low MSE across all dimensions, with errors of $6.5$, $8.8$, $11.4$, $14.3$, and $19.4$ for dimensions $2$, $3$, $4$, $5$, and $6$ respectively.

\begin{figure}[!h]
    \centering
    
    \includegraphics[width=\linewidth]{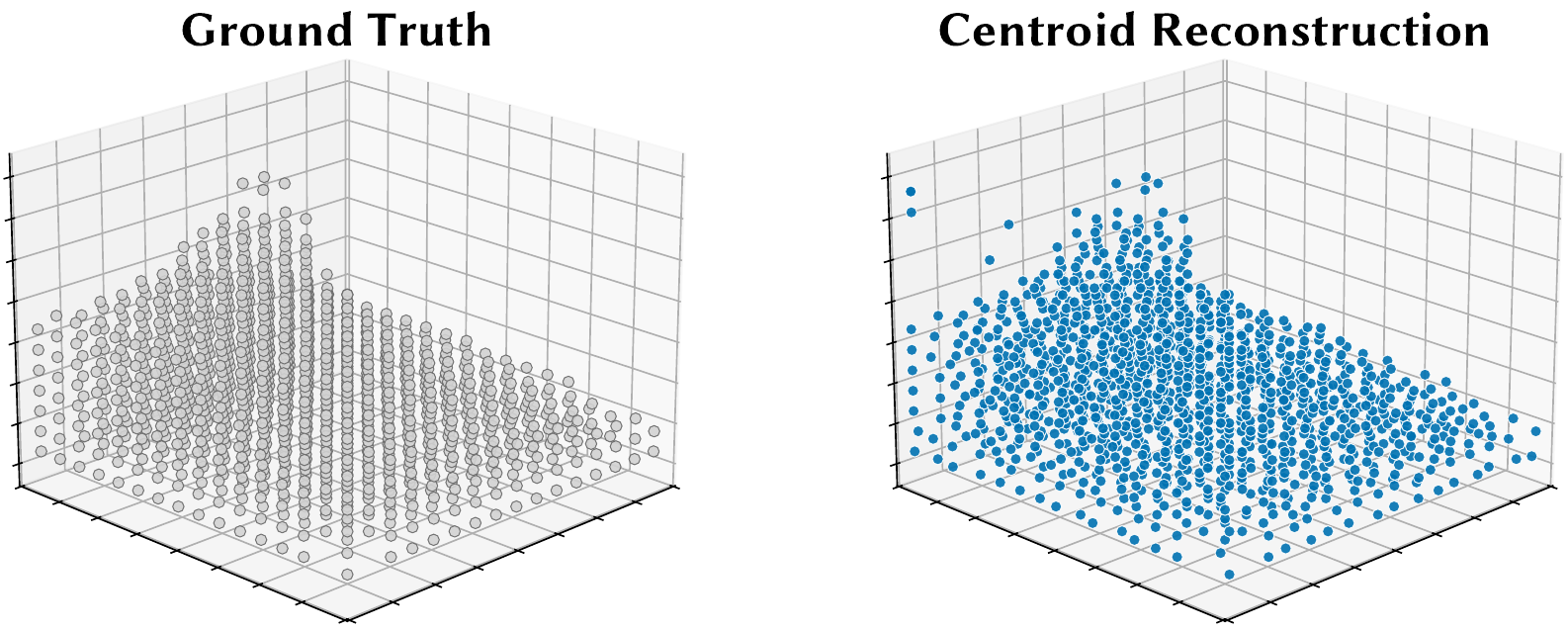}
    \caption{Reconstruction of the New Hampshire Mountains \cite{new_hampshire} using centroids over convex hulls at query percent 30 and a Gaussian distribution. MSE was 32.42}
    \label{fig:high-d}
\end{figure}

\textbf{Evaluation on Harder Query Distributions.} Figure~\ref{fig:flat} reports the cardinality of the reconstruction space of \framework under Gaussian, Uniform, and the Flattened query distribution introduced in Section~\ref{sec:flatten_tuples}, across Paris, Shanghai, and Manhattan. Under Gaussian, \framework produces a single solution with zero MSE, confirming perfect reconstruction. Under Uniform, the reconstruction space grows slightly but remains small and concentrated at low MSE. Under the Flattened distribution, the reconstruction space expands \emph{by orders of magnitude}, since this distribution has the inherent property that all $L^1$-equidistant plaintext pairs have the same probability of being retrieved.
This empirically confirms that even an information-theoretically optimal attack, i.e., \framework, can be made to stumble when the query distribution is deliberately designed to maximize reconstruction uncertainty.

\begin{figure}[htbp]
    \centering
    
    \includegraphics[width=0.82\linewidth]{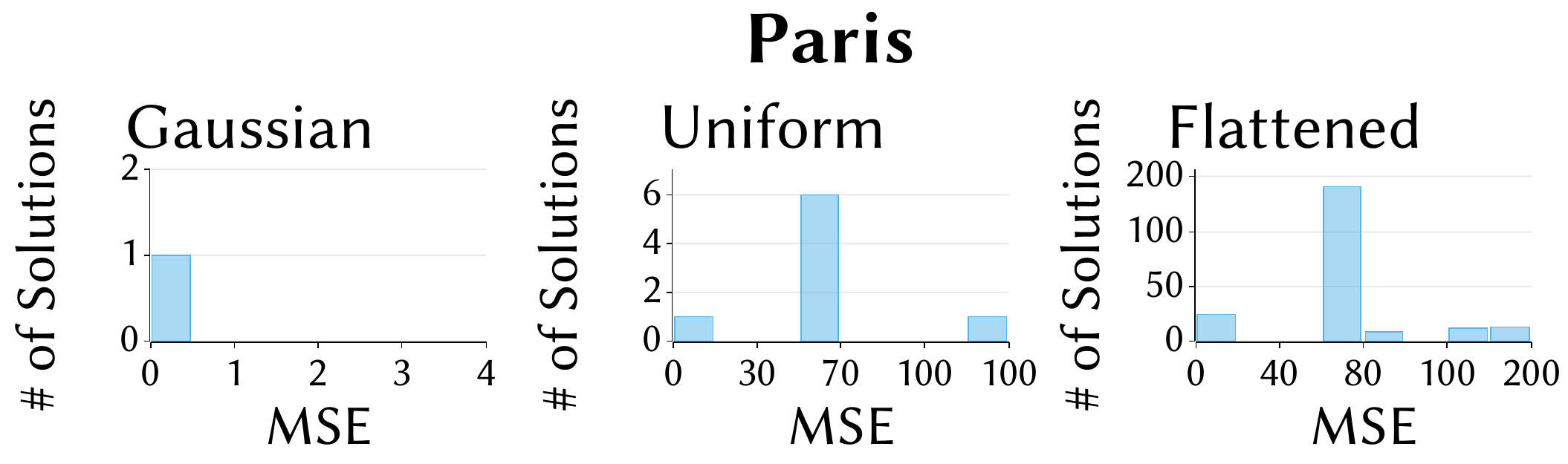}
      \includegraphics[width=0.82\linewidth]{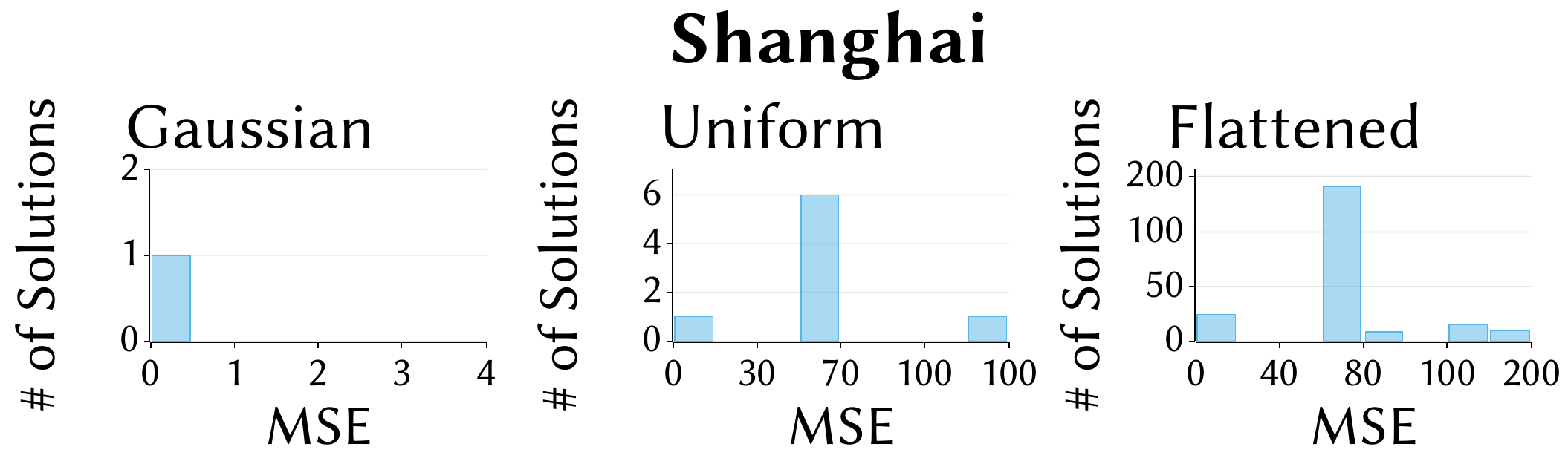}
   \includegraphics[width=0.82\linewidth]{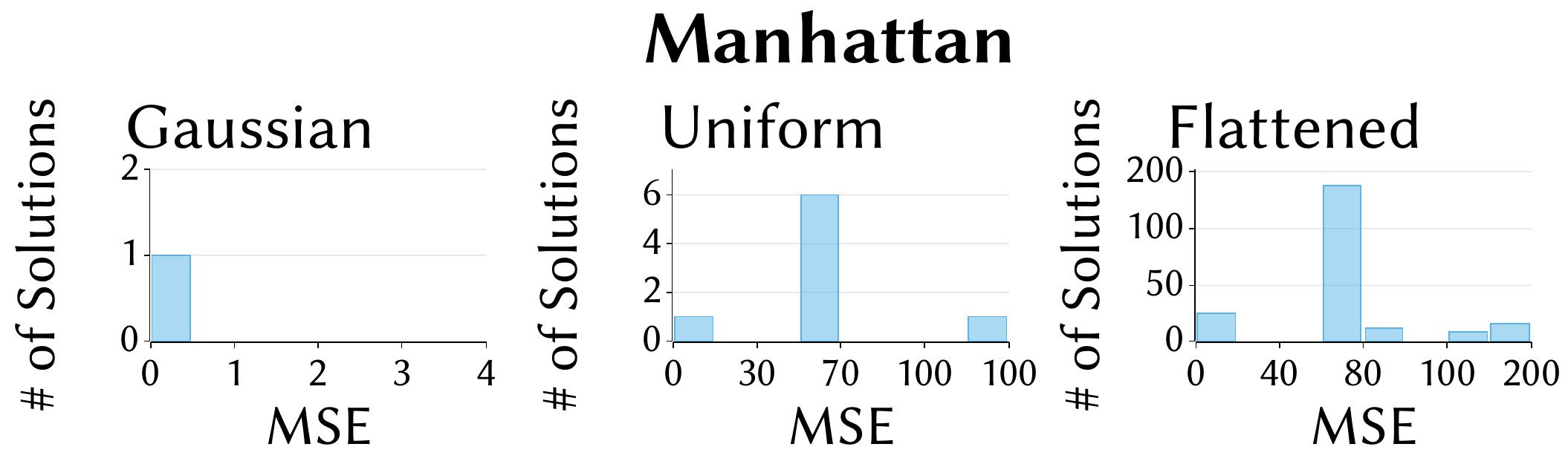}

    \caption{Evaluation of \framework under Gaussian, Uniform, and (the newly introduced) ``Flattened'' query distribution from Section~\ref{sec:flatten_tuples} that increases the reconstruction space.}
    \label{fig:flat}
\end{figure}

 \section{Conclusion}

We presented \framework, a frequency-matching attack that reconstructs plaintext coordinates of multidimensional encrypted range query databases with provable guarantees, \emph{the first of its kind}. 
By revisiting and formalizing the classical idea of frequency matching, \framework achieves coordinate-level reconstruction in arbitrary dimensions without post-hoc transformations or data injection, operating under a strictly weaker threat model than prior work~\cite{remin}. 
We complemented \framework with rigorous guarantees on query complexity, optimal parameterization via the $T_{2k}$ Selector, and worst-case output quality, and showed that these guarantees translate directly into empirical improvements across all evaluated datasets and query distributions. 
We also characterized the fundamental limits of reconstruction under query-distribution uncertainty, constructing the first formally characterized reconstruction space in this setting. Together, these results establish frequency matching as the principled foundation for multidimensional leakage cryptanalysis.

\bibliographystyle{ACM-Reference-Format} 
\bibliography{biblio}
\appendix

\section{Generative AI Usage}

We used Claude and Grammarly during the preparation
of this work to assist with phrasing and grammar edits in the main text. 
All AI-suggested text was edited and reviewed by the
authors before inclusion. 
The authors take full
responsibility for the correctness of all claims, results, and
references in this paper.

\section{Omitted Algorithms}
\label{sec:appendix_algo}

\newpage

\begin{algorithm}[t]
        \small
	 	\caption{\label{algo:singleton} Flatten Probability of Single Values }
        \KwData{Input $\qdist$ is seen as a dictionary that maps queries $\query \in \queries$ to weights $\qdist[\query]=w_{\query}$}.
        Define $v_{\textsf{mx}}$ 
        as $v_{\textsf{mx}}=\argmax_{v\in\V}\Pr[v]$ and call $\Pr[v_{\textsf{mx}}]$ as $p_1^*$\;
        Find the sum $s_{\textsf{mx}}$ of the weights of queries covering $v_{\textsf{mx}}$\;
        \For{every value $v_i$ in $\V$}
            {
            Find the sum of weights $s_{i}$ of queries covering $v_i$\;
            $\qdist[[v_i,v_i]] = \qdist[[v_i,v_i]] + (s_{\textsf{mx}} - s_{i})$\tcp*{$\Pr[v_i]= p_1^*$}
            }
        \Return $\qdist$
\end{algorithm}

\begin{algorithm}[t]
        \small
	 	\caption{\label{algo:equidistant} Flatten Frequency of Equidistant Pairs }
        \KwData{Input $\qdist$ is a dictionary that maps queries to natural number weights $\qdist[\query]=w_{\query}$}.
        \For{$d = k(N-1),k(N-1)-1,\ldots,0$}
        {
            Find the pair of distance $d$ values $t_{\textsf{mx}}=(v,v')$ with the highest probability $\Pr[v \cap v']$ among distance $d$ pairs\; 
            Find the sum $s_{\textsf{mx}}$ of weights of queries covering $t_{\textsf{mx}}$\;
            \For{every pair of values $t = (v,v')$ with distance $d$}
                {
                Find the sum of weights $s_t$ of all queries covering $t$\;
                Find the query $\query=[a,b]$ that covers $t$ and $a$ is as large as possible while $b$ is as small as possible, i.e., the minimum bounding query\;

                $\qdist[\query] = \qdist[\query] + (s_{\textsf{mx}} - s_t)$\;
                }}
        \Return{\qdist}
\end{algorithm}

\section{Omitted Proofs}
\label{appendix:proofs}

\subsection*{Proof of Lemma~\ref{lemma:covering}}

\begin{proof}
For any $V \subseteq [N]^k$, each dimension $i \in [k]$ admits a minimum and a maximum among the projections of $V$ onto that dimension. Define $V^*$ by including, for each dimension $i$, the point in $V$ attaining the minimum and the point attaining the maximum along dimension $i$; this yields $|V^*| \leq 2k$ distinct points. 
Geometrically, these $2k$ points lie on the faces of the axis-aligned hyper-rectangle they define, with all remaining points of $V$ contained in its interior or on its boundary. 
Thus, any query covering all points in $V^*$ must contain this hyper-rectangle entirely and cover all of $V$.
\end{proof}


\subsection*{Proof of Theorem~\ref{theorem:2ktuples}}

\begin{proof}
    Assume, for the sake of contradiction, that $\rdist \neq \rdist'$.
    Since the response distributions differ, then there must be at least one response $\resp \in \responses$ whose probability in $\rdist$ is different from its probability in $\rdist'$, i.e., $\Pr_{\rdist}[\resp]\neq\Pr_{\rdist'}[\resp]$.
    
    Let $\resp_{\textsf{lrg}}$ be the response with the largest number of records whose probability differs in $\rdist$ and $\rdist'$.
    Notice that $\resp_{\textsf{lrg}}$ cannot be the empty response, since if $\rdist$ and $\rdist'$ differ on the empty $\resp_{\textsf{lrg}}$, then they must differ on at least one other response $\resp'$, which will be non-empty. 
    In that case, $|\resp'|>|\resp_{\textsf{lrg}}|$ which contradicts the fact that $\resp_{\textsf{lrg}}$ has the largest number of retrieved records.
    
    The set of values in $\values$ associated with the records in $\resp_{\textsf{lrg}}$ are referred to as $V_{\textsf{lrg}}$. 
    Assume w.l.o.g. that the probability of $\resp_{\textsf{lrg}}$ in $\rdist$ is greater than its probability in $\rdist'$, that is:   
    \begin{equation}
    \Pr_{\rdist}[\resp_{\textsf{lrg}}]>\Pr_{\rdist'}[\resp_{\textsf{lrg}}].
    \label{eq:assumption_in_proof}
    \end{equation}
    By Lemma \ref{lemma:covering}, there exists a subset of records $\resp^*_{\textsf{lrg}}\subseteq \resp_{\textsf{lrg}}$ whose values comprise a set $V^*_{\textsf{lrg}}$ such that $V^{*}_{\textsf{lrg}} \subseteq V_{\textsf{lrg}}$ and $|V^{*}_{\textsf{lrg}}| \leq 2k$. 
    Specifically, set $V_{\textsf{lrg}}^*$ contains the minimal and maximal values (among the values associated with $\resp_{\textsf{lrg}}$) for each dimension.
    
    Given that, $\db'$ satisfies the constraint programming induced by $T_{2k}$, we can infer that the records associated with the values in $V^*_{\textsf{lrg}}$ form 
   a subset  $\bigcap_{\rid \in \resp^*_{\textsf{lrg}}}\rid$ which is \emph{explicitly} encoded in \framework via $T_{2k}$. 
   Thus, since both $\db$ and $\db'$ satisfy the constraints induced by subsets $T_{2k}$ and since $\bigcap_{\rid \in \resp^*_{\textsf{lrg}}}\rid$ is one of them, it must be that the frequency of this  simultaneous retrieval is equal in $\rdist$ and $\rdist'$, that is:
 \begin{equation}
\label{eq:frequencies_equal}
f\!_{\db}\left(\bigcap_{\rid \in \resp^*_{\textsf{lrg}}} \rid\right) = f\!_{\db'}\left(\bigcap_{\rid \in \resp^*_{\textsf{lrg}}} \rid\right),
\end{equation}

where $f\!_{\db}$ inidcates that the frequency is calculated over the $\db$ scenario while $f\!_{\db'}$ inidcates that the frequency is calculated over the $\db'$ scenario. 
Since we are operating in the limit, each frequency of Equation (\ref{eq:frequencies_equal}) can be written as a sum of the probability of all responses that contain all records from $\resp^*_{\textsf{lrg}}$.

    Notice that the response $\resp_{\textsf{lrg}}$ will contribute to both the expansion of the left and right frequency terms of (\ref{eq:frequencies_equal}). 
But given the inequality $\Pr_{\rdist}[\resp_{\textsf{lrg}}]>\Pr_{\rdist'}[\resp_{\textsf{lrg}}]$, the contribution of term $\Pr_{\rdist}[\resp_{\textsf{lrg}}]$ is larger  towards $f\!_{\db}\left(\bigcap_{\rid \in \resp^*_{\textsf{lrg}}} \rid\right)$ than its counterpart $\Pr_{\rdist'}[\resp_{\textsf{lrg}}]$  towards $f\!_{\db'}\left(\bigcap_{\rid \in \resp^*_{\textsf{lrg}}} \rid\right)$. 
But since the two frequencies must be equal according to (\ref{eq:frequencies_equal}), there must be a different response, call it $\resp''$, containing all the records in $\resp^*_{\textsf{lrg}}$, such that $\resp''$ has a higher probability in $\rdist'$ than $\rdist$, i.e.,     $\Pr_{\rdist}[\resp'']<\Pr_{\rdist'}[\resp'']$ to counterbalance  $\Pr_{\rdist}[\resp_{\textsf{lrg}}]>\Pr_{\rdist'}[\resp_{\textsf{lrg}}]$ so that the summation~(\ref{eq:frequencies_equal}) is the same. 
We proceed with case analysis:

$\bullet$ \emph{Case where $|\resp''| > |\resp_{\textsf{lrg}}|$.} This contradicts the assumption that $\resp_{\textsf{lrg}}$ is the largest response whose probability differs in $\rdist$ and $\rdist'$.

$\bullet$ \emph{Case where $|\resp''| < |\resp_{\textsf{lrg}}|$.} Recall that for $\resp''$ to contribute towards $f\!_{\db'}\left(\bigcap_{\rid \in \resp^*_{\textsf{lrg}}} \rid\right)$, it has to contain all records from $\resp^*_{\textsf{lrg}}$. 
From Lemma~\ref{lemma:covering}, all responses that contain $\resp^*_{\textsf{lrg}}$ must also contain the internal values of $\resp_{\textsf{lrg}}$, thus, there can not be a response containing $\resp^*_{\textsf{lrg}}$ with smaller size than $\resp_{\textsf{lrg}}$. We can dismiss this case.

$\bullet$ \emph{Case where $|\resp''| = |\resp_{\textsf{lrg}}|$.} Recall that $\resp^*_{\textsf{lrg}}$ is the subset of records in $\resp_{\textsf{lrg}}$ which are maximal/minimal among $\resp_{\textsf{lrg}}$. 
If $\resp''$ contains these records, it must also contain the other records of $\resp_{\textsf{lrg}}$, since they are internal, which means that $\resp'' = \resp_{\textsf{lrg}}$.
If $\resp''=\resp_{\textsf{lrg}}$, then recall that we already assumed\footnote{We note here that if one changes the inequality of (\ref{eq:assumption_in_proof}) to go the opposite direction, this last case of the case analysis will again reach a contradiction; this time because we can not find a $\resp''$ such that $\Pr_{\rdist}[\resp''_{\textsf{lrg}}]>\Pr_{\rdist'}[\resp''_{\textsf{lrg}}]$.} from (\ref{eq:assumption_in_proof}) that $\Pr_{\rdist}[\resp_{\textsf{lrg}}]>\Pr_{\rdist'}[\resp_{\textsf{lrg}}]$ which means that it is not possible to have $\Pr_{\rdist}[\resp'']<\Pr_{\rdist'}[\resp'']$ which is what needed to balance out the sums and get (\ref{eq:frequencies_equal}) to hold.

Thus, no response containing $\resp^*_{\textsf{lrg}}$ exists which has greater probability in $\rdist'$ than $\rdist$, i.e., $\rdist = \rdist'$.
\end{proof}

\subsection*{Proof of Theorem \ref{theorem:tightbound}}
\begin{proof}
Let $k$ and $N$ be positive integers where $N \geq 6$.
Let $\db$ and $\db'$ be databases over the set of records $\rids = \rid_1,\ldots,\rid_{2k}$ and over the domain $\V=N^k$.
Let $\qdist$ be (initially) the uniform distribution over $\queries$, the query universe for domain $\V$.
For ease of exposition, we represent the query distribution as an assignment of weights to queries: each query $\query_i$ has a positive integer weight $w_i$, and the probability of query $\query_i$ is given by its weight divided by the sum of all query weights, i.e. $\frac{w_i}{\Sigma_{j \in [|\queries|]}w_j}$.
We initialize the distribution to be uniform by requiring that $w_i=\alpha$ for all $i \in [|\queries|]$, where $\alpha$ is a positive integer.

Next, we define two disjoint sets of values, $S=\{s_1,\ldots,s_{2k}\} \subset \V$ and $S'=\{s'_1,\ldots,s'_{2k}\} \subset \V$.
Intuitively, $S$ will be the values assigned to the records in database $\db$, and $S'$ will be the values assigned to the records in database $\db'$, so that $\db(\rid_i)=s_i$ and $\db'(\rid_i)=s'_i$.

Now we define the values in $S$ and $S'$, thereby defining the values of records in each database.
For $i \in [n]$, let $s_i$ be the $k$-vector with 1 on all dimensions, except for dimension $\ceil{\frac{i}{2}}$, on which it takes value 1 if $i$ is odd and 3 if $i$ is even.
For $i \in [n]$, let $s'_i$ be the $k$-vector with $N-1$ on all dimensions, except for dimension $\ceil{\frac{i}{2}}$, on which it takes value $N$ if $i$ is odd and $N-2$ if $i$ is even.
For every value $s_i \in S$, we say that it has a corresponding value $s'_i \in S'$, and that every subset $\{s_a,s_b,\ldots,s_c\} \subset S$ has a corresponding subset $\{s'_a,s'_b,\ldots,s'_c\} \subset S'$.
Note that both databases have the following interesting property: for $i \in [2k]$, each set of $i$ records in $\db$ (resp. $\db'$) has a minimum-bounding query (MBQ) that covers no other records in $\db$ (resp. $\db'$).

Next, we adjust the query distribution by altering the weights of queries, according to the following procedure: For $i \in [2k]$, if $i$ is odd, \textbf{decrease} the weight of the MBQ of every $i$-sized subset of $S$ by $\delta$, and if $i$ is even, \textbf{increase} the weight of the MBQ of every $i$-sized subset of $S$ by $\delta$.

At the end of this procedure, the following holds:
\begin{enumerate}
    \item The probability of simultaneously querying all values in $S$, is $\delta$ greater than the probability of simultaneously querying all values in $S'$, i.e. $\Pr[S] > \Pr[S']$.
    \item The probability of simultaneously querying any strict subset of $S$ is equal to the probability of simultaneously querying its corresponding subset in $S'$, i.e. $\Pr[S_*] = \Pr[S'_*]$ for all strict non-empty subsets $S_* \subset S$.
\end{enumerate}

To see that the first statement is true, notice that the set of queries that cover $S$ and the set of queries that cover $S'$ begin with a uniform weighting, and that the only such query altered by the procedure is the MBQ of $S$.
Crucially, this query does not cover $S'$ since, by construction, every value in $S'$ dominates every value in $S$, which means that their MBQ's do not cover any of the same values.

To see that the second statement is true, consider a strict, non-empty subset $S_* \subset S$.
For each subset of $S$ containing $S_*$, the corresponding MBQ is increased by $\delta$ if the subset's size is even, and decreased by $\delta$ if its size is odd.
For any $i$ between $|S_*|$ and $|S|$, there are $\binom{|S|-|S_*|}{i}$ subsets of size $i$ that contain $S$.
Thus, the total weight added to queries covering $S_*$ is given by
\[\delta\sum_{i=0}^{|S|-|S_*|} \binom{|S|-|S_*|}{i}(-1)^{i}\]
which equals zero:
\[0 = \delta (1-1)^{|S|-|S_*|}\]
\[= \delta \sum_{i=0}^{|S|-|S_*|} \binom{|S|-|S_*|}{i} 1^{|S|-|S_*|-i}(-1)^i\]
\[=\delta \sum_{i=0}^{|S|-|S_*|} \binom{|S|-|S_*|}{i}(-1)^{i}.\] 
Thus, for all strict subsets $S_* \subset S$, it holds that $\Pr[S_*] = \Pr[S'_*]$.

To complete the proof, we assume that $\db$ is the real database, over which the leakage occurs.
The database $\db'$ \textbf{will not satisfy} the constraints induced by $T_{2k}$, since $\db$ and $\db'$ differ on $f(\rid_1 \cap \rid_2 \cap \ldots \cap \rid_{2k})$, which will be checked by $T_{2k}$'s $2k$-sized expression.
However, $\db'$ \textbf{will} satisfy the constraints induced by $T_{2k-1}$, since every set containing strictly fewer than $2k$ records will have the same frequency of being simultaneously queried in $\db$ and $\db'$, and $T_{2k-1}$ has no expressions larger than $2k-1$.

\end{proof}


\subsection*{Proof of Theorem \ref{thm:vc}}

Before proceeding to the formal proof of Theorem \ref{thm:vc}, we build intuition regarding what is required to shatter a set of responses and (in the process of building intuition) prove a Lemma used for the main proof.
Consider an example universe of five encrypted records, $\rids=\{\rid_1,\rid_2,\rid_3,\rid_4,\rid_5\}$. 
Suppose our observable responses are $\responses_\db$, and we want to see if a subset of two responses, $S=\{\resp_1,\resp_2\}\subseteq \responses_\db$, can be shattered. Shattering requires us to produce all possible subsets of $S: \emptyset, \{\resp_1\}, \{\resp_2\}, \text{and} \{\resp_1,\resp_2\}$.

Consider $\resp_1=\{\rid_1,\rid_2,\rid_3\}$ and $\resp_2=\{\rid_2,\rid_3,\rid_4\}$. 
We shatter $S$ as follows:

\begin{itemize}
    \item To isolate $\emptyset$, we need an $I$ present in neither. For example, $I=\{\rid_5\}$. $T(\{\rid_5\}) \cap S = \emptyset$.
    \item To isolate $\{\resp_1\}$, we need an $I$ in $\resp_1$ but not $\resp_2$. For example, $I=\{\rid_1\}$. $T(\{\rid_1\}) \cap S = \{\resp_1\}$.
    \item To isolate $\{\resp_2\}$, we need an $I$ in $\resp_2$ but not $\resp_1$. For example, $I=\{\rid_4\}$. $T(\{\rid_4\}) \cap S = \{\resp_2\}$.
    \item To isolate $\{\resp_1, \resp_2\}$, we need an $I$ in both. For example, $I=\{\rid_2\}$. $T(\{\rid_2\}) \cap S = \{\resp_1, \resp_2\}$.
\end{itemize}

In this case, $S$ is successfully shattered. 
An important note (that we will generalize in the upcoming Lemma) is that for each response in $\resp_i \in S$, there existed a \textit{unique} subset of encrypted records such that $T(I) \cap S = \resp_j$. 
To consider what an ``unshatterable'' set would be, consider the responses $S'=\{\resp_1, \resp_3\}$, where $\resp_1=\{\rid_1, \rid_2, \rid_3\}$ and $\resp_3=\{\rid_1, \rid_2\}$.
Notice that $\resp_3\subset \resp_1$. 
If we try to isolate the subset $\resp_3$, we must find a subset of encrypted records $I \subseteq \rids$ that is contained in $\resp_3$ but not contained in $\resp_1$. 
However, because $\resp_3$ is entirely contained in $\resp_1$, every possible subset $I \subseteq \resp_3$ is inherently also a subset of $\resp_1$. 
Therefore, any range $T(I)$ that captures $\resp_3$ will unavoidably capture $\resp_1$ as well. 
Hence, we can never produce the subset $\{\resp_3\}$ alone, and $S$ cannot be shattered.

We generalize the previous two points into the following Lemma:

\begin{lemma}
\label{lemma:shattering}
Let $S \subseteq \responses_\db$ be a set of $d$ responses shattered by a range pair $(\texttt{R}, \responses_\db)$ (Definition \ref{def:range_pair}). Then the following two properties must hold:
\begin{enumerate}
    \item For any pair of distinct responses $\resp_i, \resp_j \in S$, neither is a subset of the other ($\resp_i \not\subseteq \resp_j$ and $\resp_j \not\subseteq \resp_i$).
    \item Every individual response $\resp \in S$ must contain at least $2^{d-1}$ distinct subsets of encrypted records.
\end{enumerate}
\end{lemma}

\begin{proof}
We prove the first property first: Assume for the sake of contradiction that there exist $\resp_i, \resp_j \in S$ such that $\resp_i \subseteq \resp_j$. 
To shatter $S$, we must be able to isolate the subset $\{\resp_i\}$ without including $\resp_j$. 
This requires finding a subset of encrypted records $I$ such that $\resp_i \in T_{\responses_\db}(I)$ and $\resp_j \notin T_{\responses_\db}(I)$. 
By the definition of our range space, $\resp \in T_{\responses_\db}(I) \iff I \subseteq \resp$. 
Therefore, we need $I \subseteq \resp_i$ and $I \not\subseteq \resp_j$. 
However, since $\resp_i \subseteq \resp_j$, any subset of encrypted records $I$ contained in $\resp_i$ is necessarily also contained in $\resp_j$. 
Thus, it is impossible to create $\{\resp_i\}$, contradicting the premise that $S$ is shattered.

To prove the second property, we focus on a single response $\resp \in S$. 
Since $|S| = d$ and $S$ is shattered, the projection of $\resp$ on $S$ yields all $2^d$ possible subsets of $S$. Exactly half of these subsets, or $2^{d-1}$, must include the response $\resp$. 
Let $K \subseteq S$ be one of these $2^{d-1}$ subsets where $\resp \in K$. 
To successfully isolate $K$, there must exist a unique subset of encrypted records $I$ such that $T_{\responses_\db}(I) \cap S = K$. 
Since $\resp \in K$, it follows that $\resp \in T_{\responses_\db}(I)$, which implies $I \subseteq \resp$. Because there are $2^{d-1}$ distinct subsets $K$ that include $\resp$, each requiring its own distinct $I$, the response $\resp$ itself must contain at least $2^{d-1}$ distinct subsets of encrypted records.
\end{proof}

The previous Lemma will help us in the next few paragraphs, where we show how a problem known as the \textit{Set Union Knapsack Problem} (SUKP) (defined below) can produce a bound on the VC-dimension of our data. On a high level, the SUKP will be helpful as it will be used to optimally pack the required 'distinct subsets' from Lemma \ref{lemma:shattering}. 
Once the SUKP finds an optimal solution, we will have a bound on the largest subset that can possibly contain that many distinct identifiers. Now we formally define the SUKP:

\begin{definition}[SUKP]
    Let $U = \{a_1, \dots, a_{\ell}\}$ be a set of elements and let $\mathcal{S} = \{A_1, \dots, A_k\}$ be a set of subsets of $U$, i.e. $A_i \subseteq U$ for $1 \le i \le k$. Each subset $A_i$, $1 \le i \le k$, has an associated non-negative profit $\rho(A_i) \in \mathbb{R}^{+}$, and each element $a_j$, $1 \le j \le \ell$ has an associated non-negative weight $w(a_j) \in \mathbb{R}^{+}$. Given a subset $\mathcal{S}' \subseteq \mathcal{S}$, we define the profit of $\mathcal{S}'$ as $P(\mathcal{S}') = \sum_{A_i \in \mathcal{S}'} \rho(A_i)$. Let $U_{\mathcal{S}'} = \bigcup_{A_i \in \mathcal{S}'} A_i$. We define the weight of $\mathcal{S}'$ as $W(\mathcal{S}') = \sum_{a_j \in U_{\mathcal{S}'}} w(a_j)$. Given a non-negative parameter $c$ that we call \textit{capacity}, the \textit{Set-Union Knapsack Problem (SUKP)} requires to find the set $\mathcal{S}^* \subseteq \mathcal{S}$ which maximizes $P(\mathcal{S}')$ over all sets $\mathcal{S}'$ such that $W(\mathcal{S}') \le c$.
\end{definition} 

We explain how our terminology maps to the SUKP:

\begin{itemize}
    \item \textbf{The set of subsets $\mathcal{S}$} consists of all the unique subsets of encrypted records that appear within the responses in and observed sample $\multiset$. Every subset has a unit profit of $p(A_i) = 1$.
    \item \textbf{The universe of elements $U$} consists of the unique encrypted records that make up the responses in $\multiset$. Every record has a unit weight of $w(u) = 1$.
    \item \textbf{The capacity $c$} is set to $\ell$, where $\ell = |\{a \in \rids \mid \exists A \in M \text{ s.t. } a \in A\}|$. Simply put, $\ell$ is the total number of unique records present across our entire multiset of responses.
\end{itemize}

We see that the SUKP asks: What is the maximum number of these subsets we can select such that the union of their records contains no more than $\ell$ unique records? The optimal profit $q$ is the total count of those subsets containing unique records. This establishes $q$ as the absolute limit on how many distinct subsets of records any single response bounded by length $\ell$ could theoretically hold (and hence we use the second point from Lemma \ref{lemma:shattering} to achieve our bound).

We provide the formal proof:

\begin{proof}
Let $d = \mathrm{VC}(\texttt{R})$. By definition, this implies there exists a set of observed responses $S \subseteq M$ of cardinality $d$ that is shattered by $\texttt{R}$.

By Property 2 of Lemma \ref{lemma:shattering}, an arbitrary response $\resp \in S$ must physically contain at least $2^{d-1}$ distinct subsets of encrypted records to successfully isolate the necessary shattered subsets.

Let $\ell$ be the total number of unique encrypted records present across $\multiset$. Since the shattered response $\resp$ is drawn from $\multiset$, the total number of unique records comprising $\resp$ cannot exceed $\ell$. 

Let $q$ be the optimal profit of the SUKP associated with $\multiset$ and capacity $\ell$. As established, the optimal profit $q$ represents the absolute maximum number of distinct subsets of encrypted records that any single response of length at most $\ell$ can possibly contain. 

Because the combinatorially required number of subsets ($2^{d-1}$) cannot exceed the maximum possible physical capacity ($q$), we obtain the inequality:
\[ 2^{d-1} \le q \]

Taking the base-2 logarithm of both sides yields:
\[ d - 1 \le \log_2 q \]
\[ d \le \log_2 q + 1 \]

Since the VC-dimension $d$ represents the cardinality of a set, it must be an integer. We can therefore tighten the bound using the floor function:
\[ d \le \lfloor \log_2 q \rfloor + 1 \]

Letting $b = \lfloor \log_2 q \rfloor + 1$, we conclude that $\mathrm{VC}(\texttt{R}) \le b$.
\end{proof}

\subsection*{Proof of Theorem~\ref{theorem:flat_singleton}}

\begin{proof}
    First, Algorithm \ref{algo:singleton} identifies a value $v_{\textsf{mx}}$ with the highest probability, and computes the total weight $s_{\textsf{mx}}$ of queries covering it.
    In lines 2 through 5, it increases, for each value $v_i$, the weight of query $[v_i,v_i]$, until the total weight of queries covering $v_i$ is equal to $s_{\textsf{mx}}$.
    Since the weight $w_{q_i}$ of query $q=[v_i,v_i]$ contributes only to the probability of $v_i$, increasing $w_{q_i}$ does not affect the probabilities of other values.
    At the end of the process, although the total sum of query weights has increased, the total weights of queries covering any value $v$ is $s_{\textsf{mx}}$.
    Thus, after normalizing, we have that $\Pr_{\qdist}[e_v] = \Pr_{\qdist}[e_{v'}]$ for all $v,v' \in \V$.
\end{proof}

\subsection*{Proof of Theorem~\ref{theorem:alg2}}

\begin{proof}
    We first establish a few facts about the geometry of queries over $\V=[N]^k$.
    Recall that we define the size of a query $\query=[a,b]$ as the $L_1$ distance between $a$ and $b$, $dist(a,b)$.
    \begin{enumerate}
        \item For every pair of values $v,v'$, there is a smallest query $\query$ covering $v,v'$, called the \emph{minimum bounding query}, and the size of $q$ is equal to $dist(v,v')$.
        \item A query $q=[a,b]$ such that $dist(a,b)=d$ does not cover any pairs of values of distance greater than $d$.
    \end{enumerate}

    Algorithm \ref{algo:equidistant} iterates over all pairs of values in order of decreasing distance $d$.
    For each $d$, it finds the pair $t_{\textsf{mx}}$ of distance $d$ with the greatest probability.
    It then increases, for each pair of distance $d$, the weight of its minimum bounding query. 
    At the end of each iteration $d$, the pairs of distance $d$ are guaranteed to have the same probability.
    Furthermore, in each iteration $d$, Algorithm \ref{algo:equidistant} only alters the weights of queries with endpoints with distance $d$.
    Thus, it only alters the probabilities of pairs of values with distance $d$ or less.
    Since every iteration $d$ ends with $d$-distance pairs having the same probability, and future iterations $d-1,\ldots$ will not affect the probability of distance-$d$ pairs, it holds that Algorithm \ref{algo:equidistant} terminates by outputting a distribution $\qdist$ in which all distance-$d$ pairs have the same probability, for $d \in \{0,1,\ldots,k(N-1)\}$.
\end{proof}

\clearpage 

\section{Supplementary Plots}
\label{sec:appendix_plots}

\begin{figure}[H]
    \centering

    \includegraphics[width=0.9\linewidth]{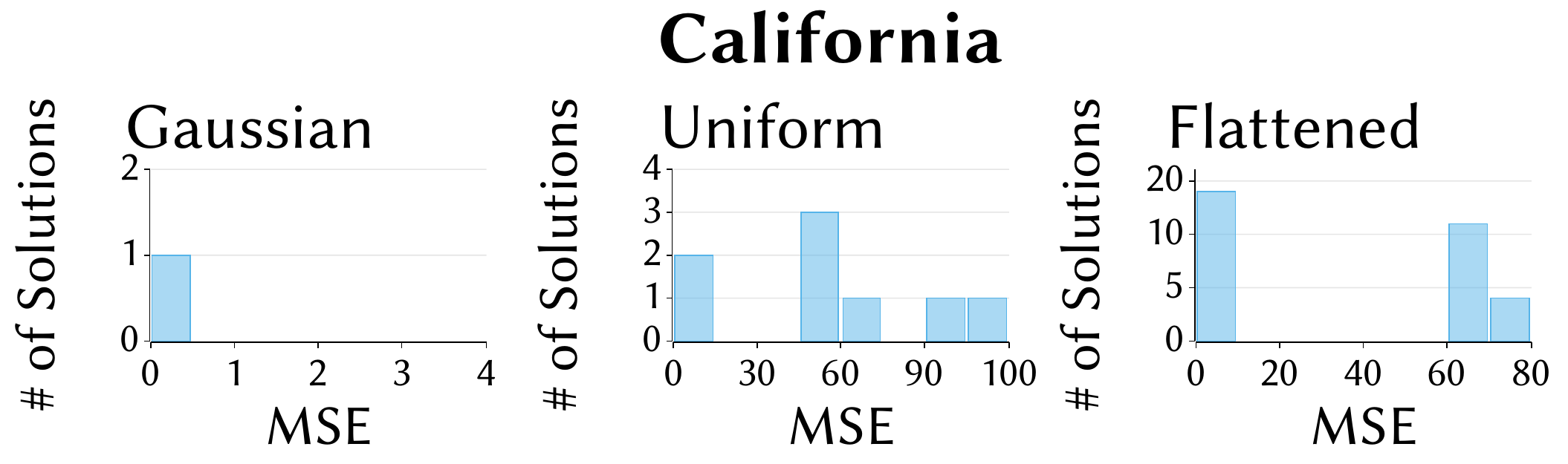}

    \includegraphics[width=0.9\linewidth]{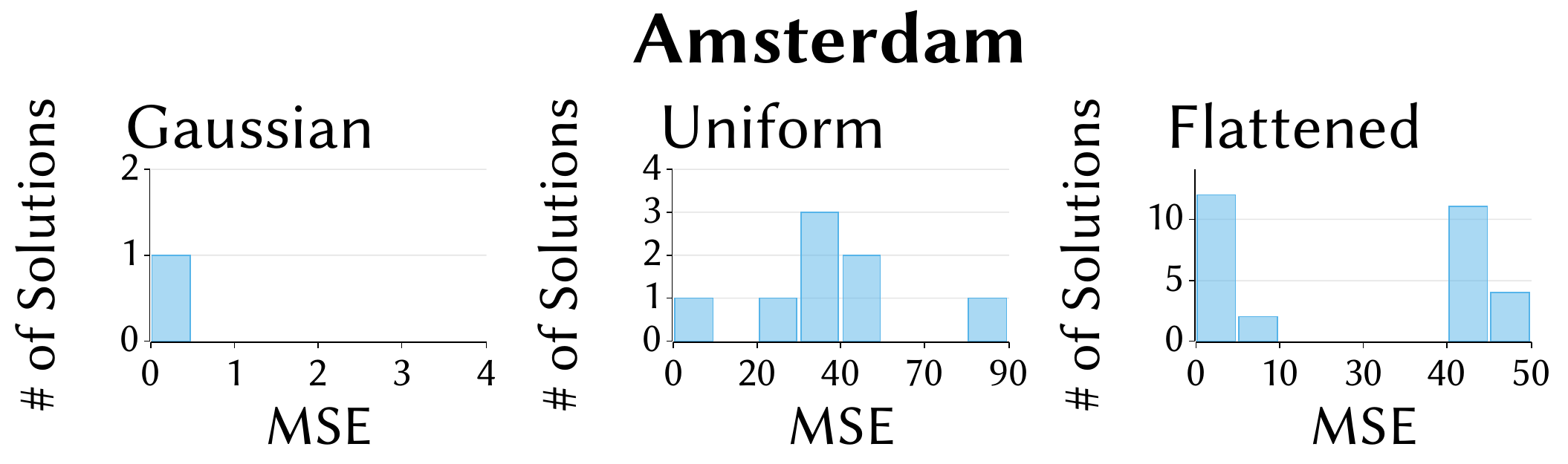}
    
    \caption{Evaluation of \framework under Gaussian, Uniform, and (the newly introduced) ``Flattened'' query distribution from Section~\ref{sec:flatten_tuples} that increases the reconstruction space.}
    \label{fig:appenidx_flat}
\end{figure}

\begin{figure}[H]
    \centering
    
  \includegraphics[width=\linewidth]{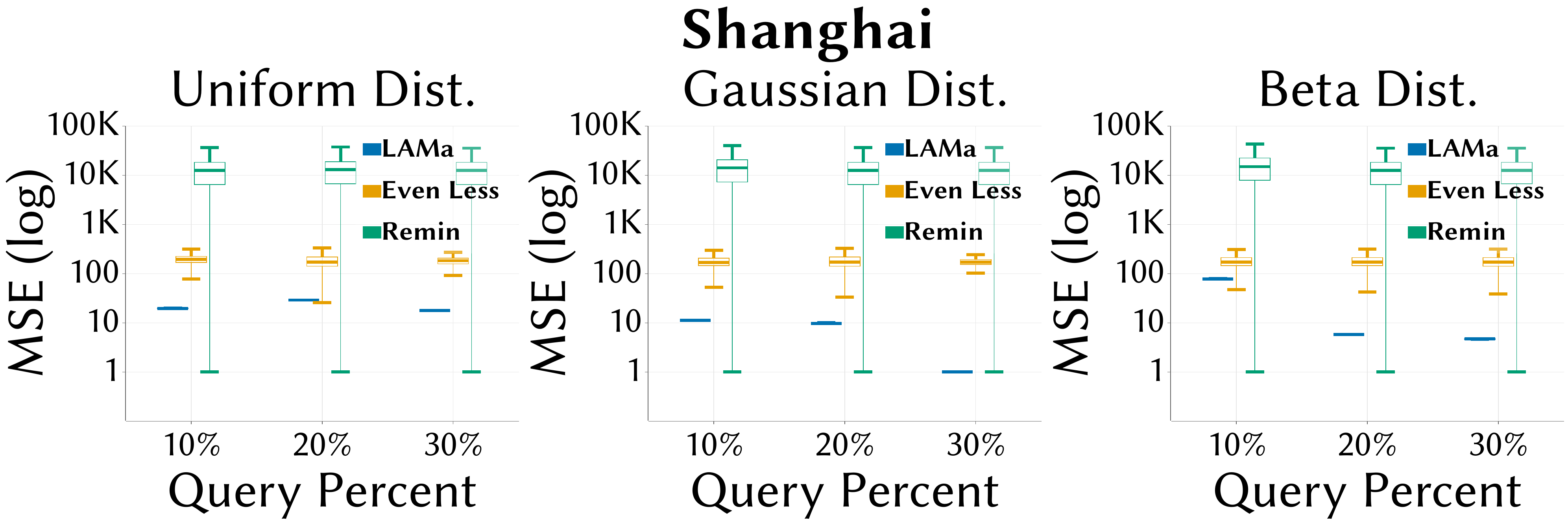}
    
     \includegraphics[width=\linewidth]{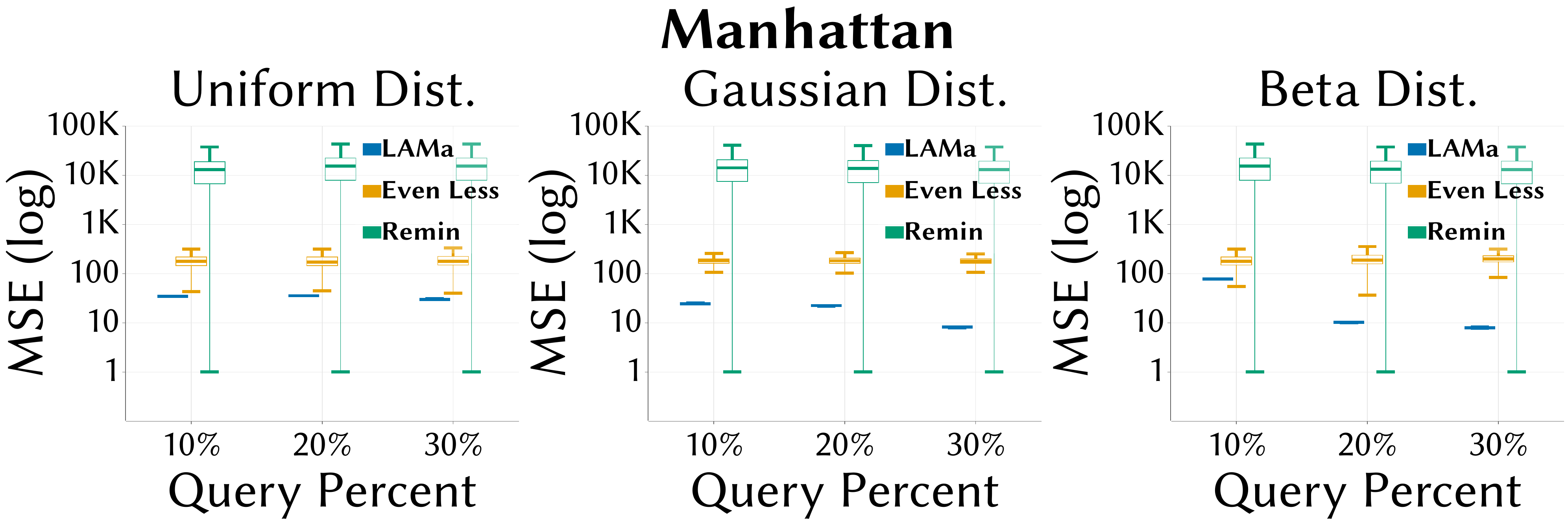}

        \caption{Solutions sampled from the total reconstruction space at each query percentage.}
    \label{fig:appendix_big_combined_figure}
\end{figure}

\begin{figure}[H]

    \centering
    
     \includegraphics[width=\linewidth]{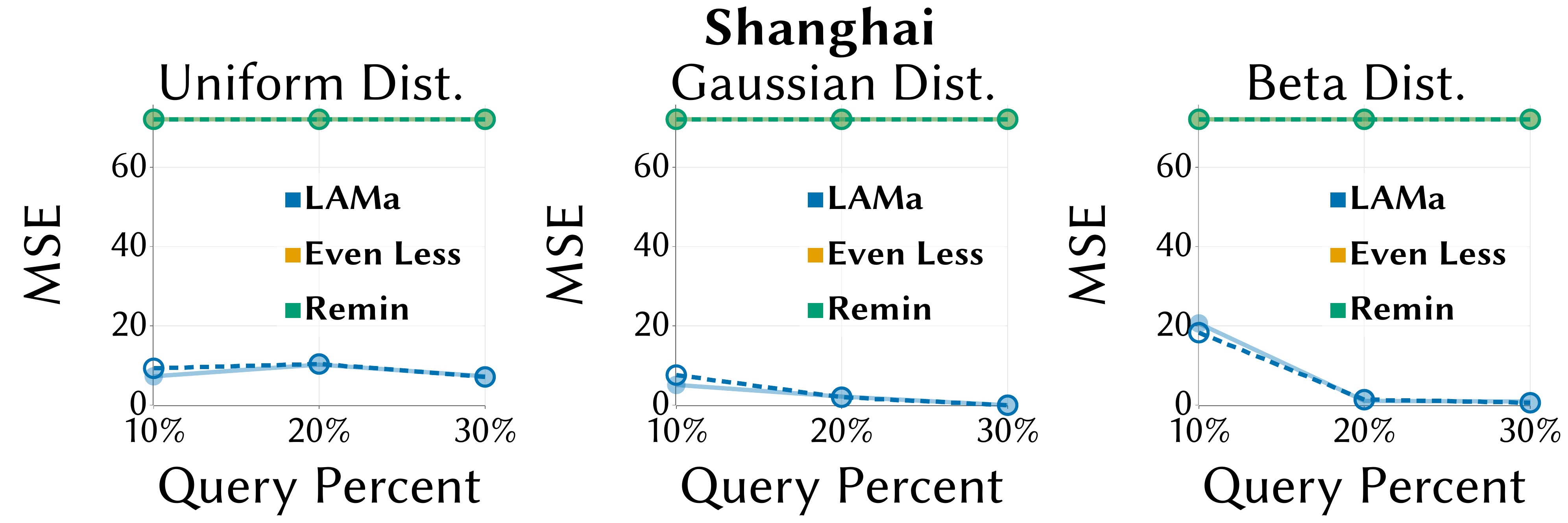}
    
     \includegraphics[width=\linewidth]{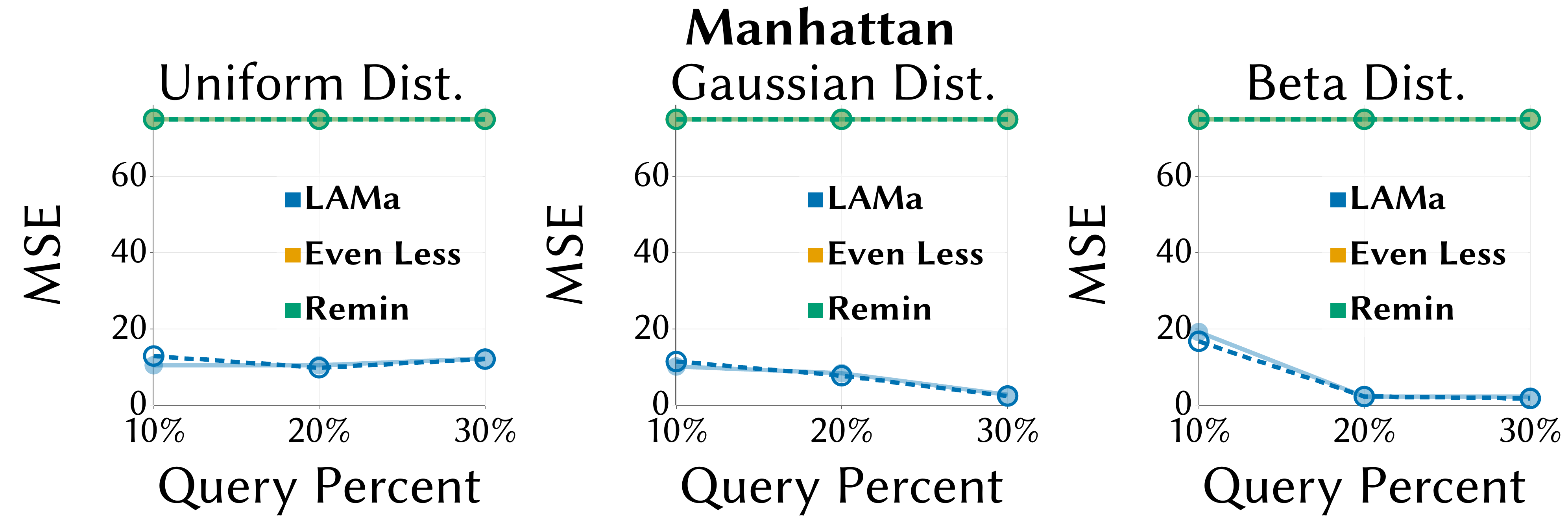}

    \caption{MSE of "worst-case" points as defined by a convex hull over the reconstruction space. Dashed line represents points placed in the middle of the largest diameter, regular line represents points placed in the barycenter.}
    \label{fig:appendix_big_combined_figure}
\end{figure}

\end{document}